\documentclass[11pt,twoside]{article}

\usepackage{jeffe,graphicx}
\usepackage[utf8]{inputenc}
\usepackage[margin=1in]{geometry}
\usepackage{marvosym}
\usepackage{cite}
\usepackage{verbatim}

\usepackage{microtype}
\usepackage[charter]{mathdesign}
\usepackage[semibold]{sourcesanspro}

\usepackage[mathcal]{euscript}
\usepackage{stmaryrd}

\def\arcto{\mathord\shortrightarrow}
\def\arc#1#2{#1\arcto#2}

\def\rev{\mathit{rev}}
\def\Z{\mathbb{Z}}
\def\Real{\mathbb{R}}

\def\reverse#1{\smash{\overline{#1}}}
\let\eps\varepsilon

\def\snip{\mathbin{\raisebox{0.15ex}{\rotatebox[origin=c]{60}{\Rightscissors}\!}}}
\def\snip{\mathbin{\backslash\!\!\backslash}}

\def\Gsnip{\mathord{G_{\subsnip}}}
\def\Sigmasnip{\mathord{\Sigma_{\subsnip}}}
\def\gammasnip{\mathord{\gamma_{\subsnip}}}

\def\Gsnip{G'}
\def\Fsnip{F'}
\def\Sigmasnip{\Sigma'}
\def\gammasnip{\gamma'}
\def\Hsnip{H'}

\def\cycle{\gamma}
\def\path{p}
\def\dualarc{\alpha}
\def\Cut{X}

\def\Sigmabar{\overline{\smash{\Sigma}\vphantom{t}}}

\def\Gbar{\overline{\smash{G}\vphantom{t}}}

\def\bbar{\overline{\smash{b}\vphantom{t}}}
\def\nbar{\overline{n}}
\def\gbar{\overline{g}}

\def\chibar{\overline{\chi}}

\newtheorem{theorem}{Theorem}[section]
\newtheorem{corollary}[theorem]{Corollary}
\newtheorem{lemma}[theorem]{Lemma}

\title{Minimum Cuts in Surface Graphs%
\thanks{
Portions of this work were presented in preliminary form, by different subsets of the authors, at
the 25th Annual Symposium on Computational Geometry~\cite{cen-mcshc-09}, the 22nd Annual ACM-SIAM
Symposium on Discrete Algorithms~\cite{en-mcsnc-11}, and the 24th Annual ACM-SIAM Symposium on Discrete Algorithms~\cite{efn-gmcse-12}.
}}

\author{
Erin W. Chambers%
	\thanks{Department of Computer Science,
	Saint Louis University;
	\url{erin.chambers@slu.edu}.
	Supported in part by NSF grants CCF-1054779, IIS-1319573, and CCF-1614562.
	Portions of this work were done while this author was a student at the University of Illinois at Urbana-Champaign.}
\and
Jeff Erickson%
	\thanks{Department of Computer Science,
	University of Illinois, Urbana-Champaign;
	\url{jeffe@illinois.edu}.
	Supported in part by NSF grants CCF-0915519, CCF-1408763, and DMS-0528086.
  }
\and
Kyle Fox%
	\thanks{Department of Computer Science,
	University of Texas at Dallas;
	\url{kyle.fox@utdallas.edu}.
	Supported in part by the Department of Energy Office of Science Graduate Fellowship Program (DOE SCGF), made possible in part by the American Recovery and Reinvestment Act of 2009, administered by ORISE-ORAU under contract no. DE-AC05-06OR23100.
	Portions of this work were done while this author was a student at the University of Illinois at Urbana-Champaign.}
\and
Amir Nayyeri%
	\thanks{School of Electrical Engineering and Computer Science,
	Oregon State University,
	\url{nayyeria@eecs.oregonstate.edu}.
	Supported in part by NSF grants CCF-1065106, CCF-0915519, and DMS-0528086.
	Portions of this work were done while this author was a student at the University of Illinois at Urbana-Champaign.}
}

\begin{document}
\maketitle
\begin{abstract}
We describe algorithms to efficiently compute minimum $(s,t)$-cuts and global minimum cuts of
undirected surface-embedded graphs.
Given an edge-weighted undirected graph $G$ with $n$ vertices embedded on an orientable surface of
genus $g$, our algorithms can solve either problem in $g^{O(g)} n \log \log n$ or $2^{O(g)} n \log
n$ time, whichever is better.
When $g$ is a constant, our $g^{O(g)} n \log \log n$ time algorithms match the best running times
known for computing minimum cuts in planar graphs.

Our algorithms for minimum cuts rely on reductions to the problem of finding a minimum-weight
subgraph in a given \(\Z_2\)-homology class, and we give efficient algorithms for this latter
problem as well.
If $G$ is embedded on a surface with $b$ boundary components, these algorithms run in $(g + b)^{O(g
+ b)} n \log \log n$ and $2^{O(g + b)} n \log n$ time.
We also prove that finding a minimum-weight subgraph homologous to a single input cycle is NP-hard,
showing it is likely impossible to improve upon the exponential dependencies on $g$ for this latter
problem.
\end{abstract}

%
%


\section{Introduction}
\label{sec:intro}


Planar graphs have been a natural focus of study for algorithms research for decades, both because
they accurately model many real-world networks, and because they often admit simpler and/or more
efficient algorithms for many problems than general graphs.  Most planar-graph algorithms either
apply immediately or have been quickly generalized to larger families of graphs, such as graphs of
higher genus, graphs with forbidden minors, or graphs with small separators.  Examples include
minimum spanning trees \cite{p-omst-99, m-tltam-04}; single-source and multiple-source shortest
paths \cite{cce-msspe-13, fr-pgnwe-06, hkrs-fspap-97, k-msspp-05, kmw-spdpg-09, lrt-gnd-79,
tm-spltm-09, efl-hmcpf-18}; graph and subgraph isomorphism \cite{g-itegd-00, hw-ltaip-74,
m-itgbg-80, e-sipgr-99, e-dtmcg-00}; and approximation algorithms for the traveling salesman
problem, Steiner trees, and other NP-hard problems~\cite{bdt-passp-14, bkk-ptass-07, bkk-stpg-07,
dhm-aacd-07, e-dtmcg-00, blw-mghls-17, gs-lsatw-02}.

The classical minimum cut problem and its dual, the maximum flow problem, are stark exceptions to this general pattern.  Flows and cuts were introduced in the 1950s as tools for studying transportation networks, which are naturally modeled as planar graphs \cite{hr-fmern-55}.  Ford and Fulkerson's seminal paper \cite{ff-mfn-56} includes an algorithm to compute maximum flows in planar networks where the source and target lie on the same face.  A long series of results eventually led to planar minimum-cut algorithms that run in near-linear time, first for undirected graphs \cite{r-mstcp-83, hj-oamfu-85, f-faspp-87, insw-iamcmf-11} and later for directed graphs \cite{jk-mcdpn-92, hkrs-fspap-97, mnnw-mdpgo-15}.

In contrast, prior to our work, almost nothing was known about computing minimum cuts in even mild generalizations of planar graphs; in particular, except for the work reported in this paper, we are unaware of any algorithm to compute minimum-cuts in non-planar graphs that does not require first computing a maximum flow.  

This paper describes the first algorithms to compute minimum cuts in surface-embedded graphs of fixed genus in near-linear time.  Specifically, we describe two algorithms to compute minimum $(s,t)$-cuts in undirected surface graphs, the first in $g^{O(g)} n\log\log n$ time, and the second in $2^{O(g)} n\log n$ time.
We also extend our  algorithms to find \emph{global} minimum cuts in undirected surface graphs in the same asymptotic time bounds.  For all our algorithms, the input consists of an $n$-vertex graph with arbitrary real edge weights, embedded on an orientable surface of genus~$g$.

Our algorithms are based on a natural generalization of the duality between cuts and cycles in planar graphs, first proposed by Whitney \cite{w-pg-33} and first exploited to compute minimum cuts in planar graphs by Itai and Shiloach \cite{is-mfpn-79}.  By definition, a set $C$ of edges defines an $(s,t)$-cut in a graph $G$ if and only if their complement $G\setminus C$ is a disconnected graph, with $s$ and $t$ in different components.  If $G$ is embedded on a surface, then the corresponding edges $C^*$ in the dual graph~$G^*$ separate the \emph{faces} of $G^*$ into two disconnected subcomplexes, one containing the dual face $s^*$ and the other containing the dual face $t^*$. 

We formalize this characterization in terms of \emph{homology}, a standard equivalence relation from algebraic topology; specifically, we use cellular homology with coefficients in $\Z_2$.  Briefly, two subgraphs of a surface graph are \emph{homologous}, or in the same \emph{homology class}, if and only if their symmetric difference is the boundary of a subset of faces.  In light of this characterization, finding minimum $(s,t)$-cuts in surface graphs becomes a special case of finding the minimum-weight subgraph of a surface graph in a given homology class.  Indeed, both of our algorithms for computing minimum $(s,t)$-cuts solve this more general problem, which is sometimes called \emph{homology localization}~\cite{cf-qhc-08,cf-hrhl-10}.

Unlike in planar graphs, where every minimal cut is  dual to a simple cycle \cite{w-pg-33}, the dual of a minimum cut in a surface graph may consist of several disjoint cycles.  More generally, the minimum-weight subgraph in any homology class may be disconnected, even when the homology class is specified by a simple cycle; see Figure \ref{fig:homology2}.  Dealing with disconnected “cycles” is a significant complication in our algorithms.

Before describing our results in further detail, we first review several related results; technical terms are more precisely defined in Section \ref{sec:prelims}.


\subsection{Past results}

\subsubsection*{Minimum cuts in planar graphs}

For any two vertices $s$ and~$t$ in a graph $G$, an \EMPH{$(s,t)$-cut} is a subset of the edges of $G$ that intersects every path from $s$ to $t$.  A \emph{minimum} $(s,t)$-cut is an $(s,t)$-cut with the smallest number of edges, or with minimum total weight if the edges of $G$ are weighted.

Itai and Shiloach \cite{is-mfpn-79} observed that the minimum $(s,t)$-cut in a planar graph~$G$ is dual to the minimum-cost cycle that separates faces $s^*$ and $t^*$ in the dual graph $G^*$.  They also observed that this separating cycle intersects any shortest path from a vertex of $s^*$ to a vertex of $t^*$ exactly once.  Thus, one can compute the minimum $(s,t)$-cut by slicing the dual graph $G^*$ along a shortest path~$\pi$ from $s^*$ to~$t^*$; duplicating every vertex and edge of $\pi$; and then computing, for each vertex $u$ of $\pi$, the shortest path between the two copies of $u$ in the resulting planar graph.  Applying Dijkstra's shortest-path algorithm at each vertex of~$\pi$ immediately yields a running time of $O(n^2\log n)$.

Reif \cite{r-mstcp-83} improved the running time of this algorithm to $O(n\log^2 n)$ using a divide-and-conquer strategy.  Reif's algorithm was extended by Hassin and Johnson to compute the actual maximum flow in $O(n\log n)$ additional time, using a carefully structured dual shortest-path computation \cite{hj-oamfu-85}.  The running time was improved to $O(n\log n)$ by Frederickson~\cite{f-faspp-87}, and more recently to $O(n\log\log n)$ by Italiano \etal~\cite{insw-iamcmf-11}, by using a balanced separator decomposition to speed up the shortest-path computations.

Janiga and Koubek~\cite{jk-mcdpn-92} attempted to adapt Reif's $O(n\log^2 n)$-time algorithm to directed planar graphs; however, their  algorithm has a subtle error~\cite{kn-mcupg-11} which may lead to an incorrect result when the minimum $(t,s)$-cut is smaller than the minimum $(s,t)$-cut.

Henzinger \etal~\cite{hkrs-fspap-97} generalized Frederickson's technique to obtain an $O(n)$-time
planar shortest-path algorithm; using this algorithm in place of Dijkstra's algorithm improves the
running times of both Reif's and Janiga and Koubek's algorithms to $O(n\log n)$.  The same
improvement can also be obtained using more recent multiple-source shortest path algorithms by
Klein~\cite{k-msspp-05}; Cabello, Chambers, and Erickson~\cite{cce-msspe-13}; and Erickson, Fox, and
Lkhamsuren~\cite{efl-hmcpf-18}.

Minimum $(s,t)$-cuts in directed planar graphs can also be computed in $O(n\log n)$ time using the planar maximum-flow algorithms of Weihe \cite{w-mstfp-97} (after filtering out useless edges \cite{fls-fuadp-18}) and Borradaile and Klein \cite{b-epnfc-08, bk-tamfd-06, bk-amfdp-09}.

A \EMPH{cut} (without specified $s$ and $t$) is a subset of edges of $G$ that separate $G$ into two non-empty sets of vertices.
A \emph{global minimum} cut is a cut of minimum size, or minimum total weight if the edges of~$G$
are weighted.  Equivalently, a global minimum cut is an $(s,t)$-minimum cut of smallest total
weight, minimized over all pairs of vertices $s$ and $t$.  Chalermsook, Fakcharoenphol, and
Nanongkai~\cite{cfn-dnlta-04} gave the first algorithm for computing global minimum cuts that relies
on planarity; their algorithm runs in~$O(n \log^2 n)$ time.  Their algorithm was improved by
\L\c{a}cki and Sankowski~\cite{ls-mcsc-11} who achieved an $O(n \log \log n)$ running time.  Mozes
\etal~recently achieved the same $O(n \log \log n)$ running time for global minimum cuts in
\emph{directed} planar graphs \cite{mnnw-mdpgo-18},  using techniques reported in a preliminary
version of the current paper \cite{en-mcsnc-11}, specifically, the $\Z_2$-homology covers described in Section \ref{sec:homcover}.

\subsubsection*{Generalizations of planar graphs}

Surprisingly little is known about the complexity of computing maximum flows or minimum cuts in generalizations of planar graphs.  In particular, we know of no previous algorithm to compute minimum cuts in non-planar graphs that does not first compute a maximum flow.

By combining a technique of Miller and Naor \cite{mn-fpgms-95} with the planar directed flow
algorithm of Borradaile and Klein \cite{b-epnfc-08, bk-tamfd-06, bk-amfdp-09, e-mfpsp-10}, one can
compute maximum (single-commodity) flows in a planar graph with $k$ sources and sinks in $O(k^2
n\log n)$ time.  Very recently, Borradaile \etal~\cite{bkmnw-msmsm-17} described an algorithm to
compute maximum flows in planar graphs with an arbitrary number of sources and sinks in $O(n\log^3 n)$ time.  An algorithm of Hochstein and Weihe \cite{hw-mstfkc-07} computes maximum flows in planar graphs with $k$ additional edges in $O(k^3n\log n)$ time, using a clever simulation of Goldberg and Tarjan's push-relabel algorithm~\cite{gt-namfp-88}.  Borradaile \etal~\cite{bkmnw-msmsm-17} extend Hochstein and Weihe's framework to compute maximum flows in planar graphs with $k$ apices in $O(k^3n\log^3 n)$ time.

Chambers and Eppstein \cite{ce-focmf-13} describe an algorithm to compute maximum flows in $O(n\log n)$ time if the input graph forbids a fixed minor that can be drawn in the plane with one crossing.  Another related result is the algorithm of Hagerup \etal~\cite{hknr-cmfnc-98} to compute maximum flows in graphs of constant treewidth in $O(n)$ time.

Imai and Iwano \cite{ii-espap-90} describe a max-flow algorithm that applies to graphs of positive genus, but not to arbitrary sparse graphs.
Their algorithm computes minimum-cost flows in graphs with small balanced separators, using a combination of nested dissection \cite{lrt-gnd-79, pr-fepss-93}, interior-point methods~\cite{v-slpfm-89}, and fast matrix multiplication.
Their algorithm can be adapted to compute maximum flows (and therefore minimum cuts) in any graph of constant genus in time $O(n^{1.595}\log C)$, where $C$ is the sum of integer edge weights.
However, this algorithm is slower than more recent and more general algorithms \cite{ds-flgfi-08, gr-bfdb-98}.

Chambers, Erickson, and Nayyeri~\cite{cen-hfcc-12} describe maximum flow algorithms that are tailored specifically for graphs of constant genus.
Given a graph embedded on a surface of genus~$g$, their algorithms compute a maximum flow in  $O(g^8 n \log^2 n \log^2 C)$ time where $C$ is the sum of integer edge weights and in $g^{O(g)}n^{3/2}$ arithmetic operations when edge weights are arbitrary positive real numbers.  Their key insight is that it suffices to optimize the \emph{homology class} (with coefficients in $\Real$) of the flow, rather than directly optimizing the flow itself.

Euler's formula implies that a simple $n$-vertex graph embedded on a surface of genus $O(n)$ has at most $O(n)$ edges.
The fastest known combinatorial maximum-flow algorithm for sparse graphs, due to Orlin~\cite{o-mfotl-13}, runs in $O(n^2 / \log n)$ time.
The fastest algorithm known for sparse graphs with small integer capacities, due to Goldberg and Rao
\cite{gr-bfdb-98} and Lee and Sidford~\cite{ls-pfmlp-14}, run in time $O(n^{3/2} \polylog(n, U))$, where $U$ is an upper bound on the integer edge weights.
Madry~\cite{m-ncpef-13} describes a faster algorithm for unit capacity graphs that runs
in~$O(n^{10/7} \polylog n)$ time when the graph is sparse.

The fastest algorithm known to compute global minimum cuts in arbitrary weighted undirected graphs is a Monte Carlo randomized algorithm of Karger~\cite{k-mcnlt-00}, which runs in~$O(m \log^3 n)$ time but fails with small probability. 
A more recent deterministic algorithm of Henzinger, Rao, and Wang \cite{hrw-lfpfe-17}, based on breakthrough techniques of Kawarabayashi and Thorup \cite{kt-dgmcs-15,kt-decnt-18}, computes global minimum cuts in \emph{unweighted} graphs in $O(m \log^2 n \log^2 \log n)$ time.
The fastest deterministic algorithms known for global minimum cuts in arbitrary weighted graphs run in $O(nm + n^2\log n)$ time for undirected graphs \cite{ni-cemcg-92,f-eani-94,sw-sma-97} and in $O(mn \log(n^2/m))$ time for directed graphs \cite{ho-fafmd-94}.

For further background on maximum flows, minimum cuts, and related problems, we refer the reader to monographs by Ahuja \etal\ \cite{amo-nftaa-93} and Schrijver \cite{s-cape-03}.

\subsubsection*{Optimal homology representatives}

Homology is a topological notion of equivalence with nice algebraic properties.  Two subgraphs of a surface graph $G$ are \emph{homologous}, or in the same \emph{homology class}, if their difference is the sum of face boundaries, where summation is defined over some coefficient ring.  Our minimum-cut algorithms all reduce to the problem of finding a subgraph of \emph{minimum weight} in a given homology class (over the ring~$\Z_2$).  Several  authors have considered variants of this problem, which is often called \emph{homology localization}.

Most interesting variants of homology localization are NP-hard.  Chambers \etal~\cite{ccelw-scsih-08} prove that finding the shortest \emph{splitting} cycle is {NP}-hard; a cycle is splitting if it is non-self-crossing, non-contractible, and null-homologous.  A simple modification of their reduction (from Hamiltonian cycle in planar grid graphs) implies that finding the shortest \emph{simple cycle} in a given homology class is {NP}-hard.  Chen and Freedman \cite{cf-qhc-08, cf-qhc2-07} proved a similar hardness result for general simplicial complexes; however, the complexes output by their reduction are never manifolds.  Recently, Grochaw and Tucker-Foltz \cite{gt-ctugc-18} proved that homology localization in surface graphs, over a sufficiently large finite coefficient ring, is equivalent to Unique Games; in particular, there is no PTAS for \emph{any} finite coefficient ring unless the Unique Games Conjecture is false.

On the other hand, for homology with real or integer coefficients, homology localization in surface graphs is equivalent (via duality) to a minimum-cost flow problem and hence can be solved in polynomial time.  Chambers \etal~\cite{cen-hfcc-12} describe an algorithm to find optimal circulations in a given homology class in near-linear time, given a graph with integer coefficients on a surface of fixed genus.  Sullivan \cite{s-cath-90} and Dey \etal\ \cite{dhk-ohctu-11} prove similar results for higher-dimensional orientable manifolds.

\subsection{New results and organization}

In Section \ref{S:tree-cotree}, we describe two techniques to preprocess a graph on a surface with boundary, so that the homology class of any subgraph can be computed quickly.  These are both straightforward generalizations of standard methods for measuring homology in surfaces \emph{without} boundary based on tree-cotree decompositions \cite{ew-gohhg-05, ccelw-scsih-08, e-dgteg-03}.  In particular, we describe how to construct a \emph{system of arcs}---a collection of $O(g+b)$ boundary-to-boundary paths that cut the surface into a disk---in $O((g+b)n)$ time.  This generalization is essential for our algorithms, as our dual homology characterization of minimum $(s,t)$-cuts removes the dual faces $s^*$ and $t^*$, leaving a surface with two boundary components.

In Section~\ref{sec:crossing}, we present our first algorithm to compute minimum-weight subgraphs in a given homology class.  Our algorithm first computes a \emph{greedy system of arcs}; each arc in this system consists of two shortest paths.  Using an exchange argument, we prove that the minimum-weight subgraph in any homology class crosses each arc in the greedy system at most $O(g+b)$ times.  Our algorithm enumerates all possible sequences of crossings consistent with this upper bound, and finds the shortest subgraph consistent with each crossing sequence, by reducing to a planar minimum cut problem.  The resulting algorithm runs  in $(g+b)^{O(g+b)}n\log \log n$ time.

We describe our second algorithm to compute minimum-weight homologous subgraphs in
Section~\ref{sec:homcover}.  Instead of considering the \emph{sequence} of crossings with a
\emph{greedy} system of arcs, we instead count the \emph{number} of crossings with each arc in an
\emph{arbitrary} system of arcs.  The resulting vector of crossing numbers for an even subgraph
characterizes the homology class of that subgraph.  Our algorithm computes the shortest \emph{cycle}
in \emph{every} homology class, by constructing and searching a certain covering space of the
surface  that we call the \emph{$\Z_2$-homology cover}, using an extension~\cite{efl-hmcpf-18} of the multiple-source shortest path algorithm of Cabello \etal \cite{cce-msspe-13}.  We then assemble the minimum-weight \emph{even subgraph} in any desired homology class from these $\Z_2$-minimal cycles using dynamic programming.  The resulting algorithm runs in $2^{O(g+b)}n\log n$ time.

In Section~\ref{S:NPhard}, we prove that finding a minimum-weight even subgraph in a given homology class in NP-Hard.  
Unlike Chen and Freedman \cite{cf-hrhl-10}, this reduction is done on a 2-manifold, and unlike Chambers \etal~\cite{ccelw-scsih-08}, the target subgraph does not need to be a simple cycle.
This reduction implies that the exponential dependence on $g$ in our algorithms is unavoidable.  

Finally, in Section~\ref{sec:global}, we describe our algorithms for computing global minimum cuts.  Both algorithms ultimately reduce computing a global minimum cut to  $2^{O(g)}$ instances of computing minimum $(s,t)$-cuts; thus, our algorithms have the same asymptotic running times as the minimum $(s,t)$-cut algorithms from Sections \ref{sec:crossing} and \ref{sec:homcover}.

We note with some amusement that our algorithms solve a problem with a well-known polynomial-time solution by reducing it to an exponential number (in $g$) of instances of an NP-hard (but fixed-parameter tractable) problem!  
The authors of this paper are divided on whether to conjecture that minimum cuts in surface graphs can  be computed in time $O(g^c n\polylog n)$ for some small constant $c$, or that the problem is “fixed-parameter quadratic” with respect to genus, just as diameter and radius are fixed-parameter quadratic with respect to treewidth \cite{aww-afpsa-16}.  
Fomin \etal~\cite{flspw-fppcg-18} raise similar questions about the fixed-parameter efficiency of flows and cuts with respect to treewidth.

\section{Notation and Terminology}
\label{sec:prelims}

%



We begin by recalling several useful definitions related to surface-embedded graphs.  For further background, we refer the reader to Gross and Tucker \cite{gt-tgt-01} or Mohar and Thomassen~\cite{mt-gs-01} for topological graph theory, and to Hatcher~\cite{h-at-02} or Stillwell~\cite{s-ctcgt-93} for surface topology and homology.

\subsection{Surfaces and curves}
\label{SS:surfaces}

A \EMPH{surface} (more formally, a \emph{2-manifold with boundary}) is a compact Hausdorff space in which every point has an open neighborhood homeomorphic to either the plane $\Real^2$ or a closed halfplane $\set{(x,y)\in \Real^2\mid x\ge 0}$.  The points with halfplane neighborhoods make up the \EMPH{boundary} of the surface; every component of the boundary is homeomorphic to a circle.
A surface is \EMPH{non-orientable} if it contains a subset homeomorphic to
the M\"obius band, and \EMPH{orientable} otherwise. In this paper, we consider only compact, connected, and orientable surfaces.

A \EMPH{path} in a surface $\Sigma$ is a continuous function $p\colon [0,1]\to\Sigma$.
A \EMPH{loop} is a path whose endpoints~$p(0)$ and~$p(1)$ coincide;
we refer to this common endpoint as the \EMPH{basepoint} of the loop.
An \EMPH{arc} is a path internally disjoint from the boundary of~$\Sigma$
whose endpoints lie on the boundary of $\Sigma$.
A \EMPH{cycle} is a continuous function $\gamma\colon S^1\to\Sigma$;
the only difference between a cycle and a loop is that a loop has a
distinguished basepoint.
We say a loop~$\ell$ and a cycle~$\gamma$ are \EMPH{equivalent} if, for some
real number~$\delta$, we have~$\ell(t) = \gamma(t + \delta)$ for
all~$t \in [0,1]$.
We collectively refer to paths, loops, arcs, and cycles as \EMPH{curves}.
A~curve is \EMPH{simple} if it is injective; we usually do not distinguish between simple curves and their images in~$\Sigma$.
A simple curve~$p$ is \EMPH{separating} if~$\Sigma \setminus p$ is disconnected.

The \EMPH{reversal}~$\rev(p)$ of a path~$p$ is defined by
setting~$\rev(p)(t) = p(1-t)$. The \EMPH{concatenation}~$p \cdot q$ of two
paths~$p$ and~$q$ with~$p(1)=q(0)$ is the path created by
setting~$(p\cdot q)(t) = p(2t)$ for all~$t \leq 1/2$
and~$(p\cdot q)(t) = q(2t-1)$ for all~$t \geq 1/2$.

The \EMPH{genus} of a surface $\Sigma$ is the maximum number of disjoint simple cycles in $\Sigma$ whose complement is connected.
 Up to homeomorphism,
there is exactly one orientable surface and one non-orientable surface with any genus $g\ge 0$ and any number of
boundary cycles $b\ge 0$.
Orientable surfaces with~$b$ boundary components are differentiated by their \EMPH{Euler characteristic} ${\chi = 2 - 2g - b}$ (for non-orientable surfaces, ${\chi = 2 - g - b}$).

\subsection{Graph embeddings}
\label{SS:embeddings}

An \EMPH{embedding} of an undirected graph $G=(V,E)$ on a surface $\Sigma$ maps vertices to distinct points and edges to simple, interior-disjoint paths.  The \EMPH{faces} of the embedding are maximal connected subsets of $\Sigma$ that are disjoint from the image of the graph.
We may denote an edge~$uv \in E$ as~$f | g$ if it is incident to faces~$f$ and~$g$.
An embedding is \EMPH{cellular} if each of its faces is homeomorphic to the plane; in particular, in any cellular embedding, each component of the boundary of $\Sigma$ must be covered by a cycle of edges in~$G$.  Euler's formula implies that any cellularly embedded graph with $n$ vertices, $m$ edges, and $f$ faces lies on a surface with Euler characteristic $\chi = n-m+f$, which implies that $m = O(n+g)$ and $f=O(n+g)$
if the graph is simple.
We consider only such
cellular embeddings of genus $g=O(n^{1-\eps})$, so that the overall complexity of the embedding is $O(n)$.

Any cellular embedding on an orientable surface can be encoded combinatorially by a \EMPH{rotation system}, which records the counterclockwise order of edges incident to each vertex.
We also refer to the complex of vertices, edges, and faces induced by a cellular embedding as a \EMPH{combinatorial surface}.
Every combinatorial surface with boundary can be obtained from a combinatorial surface without boundary by deleting the interiors of one or more faces.

We redundantly use the term \EMPH{arc} to refer to a walk in the graph whose endpoints are boundary vertices.  Likewise, we use the term \EMPH{cycle} to refer to a closed walk in the graph.  Note that arcs and cycles may traverse the same vertex or edge more than once.

An arc or cycle in a combinatorial surface is \EMPH{weakly simple} if it can be continuously and
infinitesimally perturbed on the underlying 2-manifold $\Sigma$ into a simple path or cycle; we note that algorithms to detect if a cycle is weakly simple has been studied extensively of late~\cite{cex-dwsp-15,aaet-rwsp-17}.
Similarly, an arc or cycle~$\alpha$ and another arc or cycle $\beta$ are \EMPH{non-crossing} if some arbitrarily small perturbations of $\alpha$ and $\beta$ are disjoint; otherwise, we say that $\alpha$ \EMPH{crosses} $\beta$.

An \EMPH{even subgraph} is a subgraph of $G$ in which every node has even degree, or equivalently, the symmetric difference of cycles.
A \EMPH{cycle decomposition} of an even subgraph~$H$ is a set of edge-disjoint, non-crossing, weakly simple cycles whose union is $H$.

\begin{lemma}
\label{lem:decomposition}
Every even subgraph of an embedded graph has a cycle decomposition.
\end{lemma}

\begin{proof}
Let $H$ be an even subgraph of $G$.  We can decompose $H$ into cycles by specifying, at each vertex~$v$, which pairs of incident edges of $H$ are consecutive.  Any pairing that does not create a crossing at $v$ is sufficient.  For example, if $e_1, e_2, \dots, e_{2d}$ are the edges of $H$ incident to $v$, indexed in clockwise order around $v$, we could pair edges $e_{2i-1}$ and $e_{2i}$ for each $i$.
\end{proof}

We emphasize that each cycle in a cycle decomposition may visit vertices multiple times; indeed, some even subgraphs cannot be decomposed into strictly simple cycles.

\EMPH{Slicing} a combinatorial surface along a cycle or arc modifies both the surface and the embedded graph.
For any combinatorial surface $S = (\Sigma, G)$ and any simple cycle or arc $\gamma$ in~$G$, we define a new combinatorial surface \EMPH{$S \snip \gamma$} by taking the topological closure of $\Sigma \backslash \gamma$ as the new underlying surface; the new embedded graph contains two copies of each vertex and edge of $\gamma$, each bordering a new boundary.
We define the~\EMPH{projection} of a curve in~$S \snip \gamma$ as the natural mapping of points (or vertices and edges) to~$S$. 

\subsection{Duality}

Any undirected graph~$G$ embedded on a surface~$\Sigma$ without boundary has a
\EMPH{dual graph}~$G^*$, which has a vertex~$f^*$ for each face~$f$ of~$G$,
and an edge~$e^*$ for each edge~$e$ in~$G$ joining the vertices dual to the
faces of~$G$ that~$e$ separates. The dual graph~$G^*$ has a natural cellular
embedding in~$\Sigma$, whose faces corresponds to the vertices of~$G$.
See Figure~\ref{fig:prelims_primaldual}.

\begin{figure}[ht]
\centering
\includegraphics[height=1.25in]{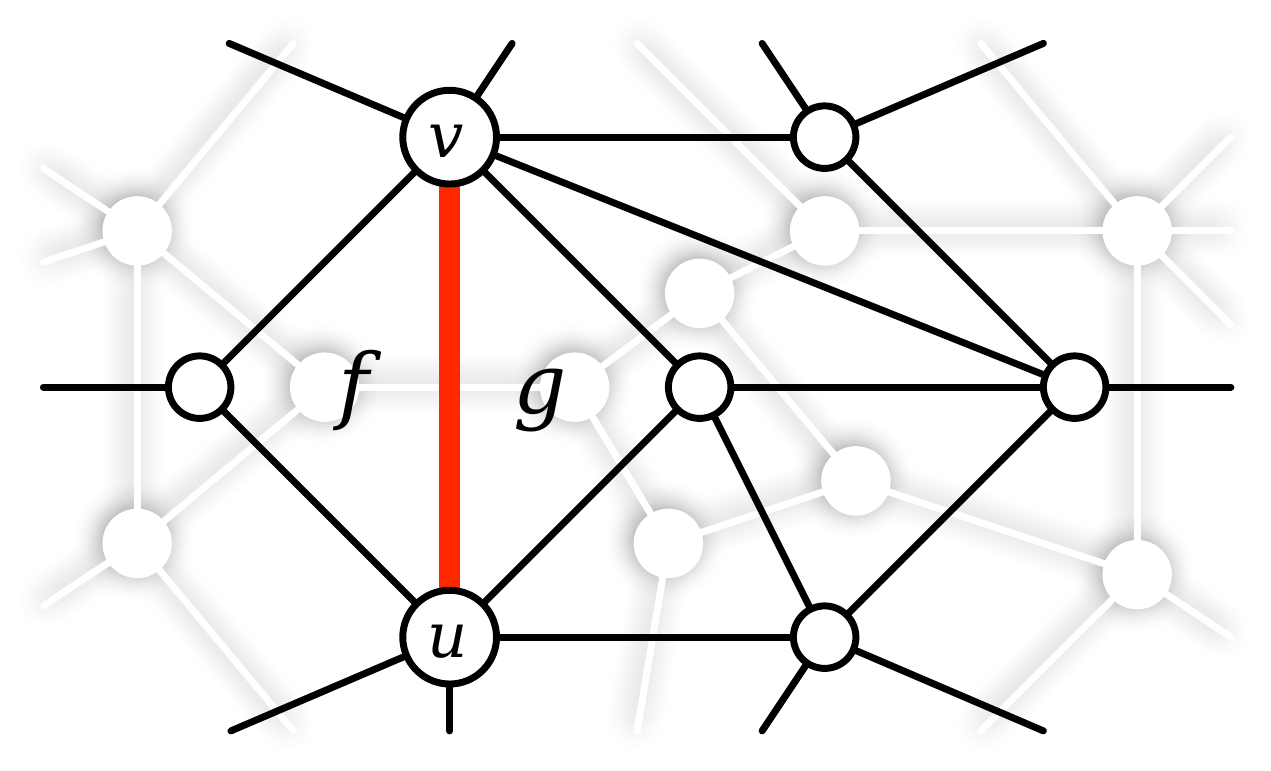}\quad
\includegraphics[height=1.25in]{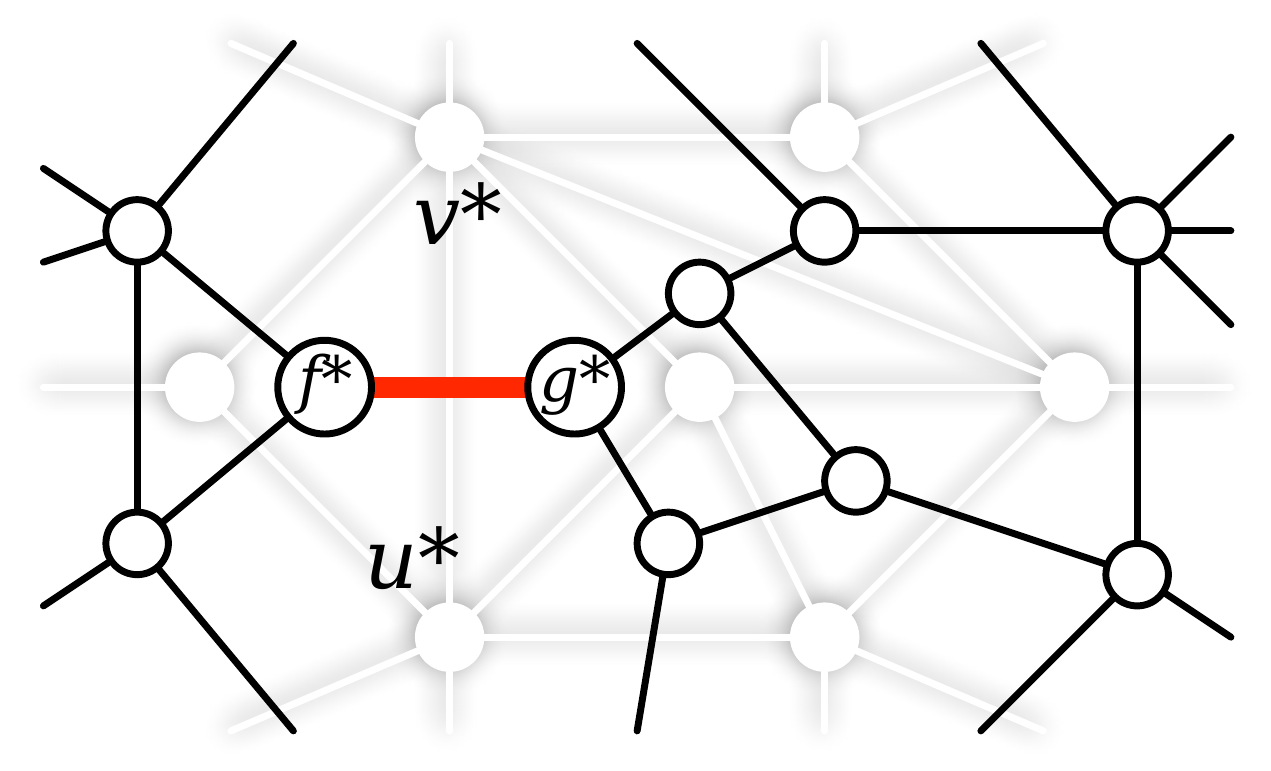}
\caption{Graph duality.  One edge $uv$ and its dual $(uv)^* =
f^*g^*$ are emphasized.} \label{fig:prelims_primaldual}
\end{figure}

Any undirected graph $G$ embedded on a surface $\Sigma$ with boundary has a \EMPH{dual graph~$G^*$}, defined as follows.\footnote{Our definition differs slightly from the one proposed by Erickson and Colin de Verdière~\cite{ce-tnpcs-10}.}  The dual graph~$G^*$ has a vertex $f^*$ for each face $f$ of~$G$, \emph{including the boundary cycles}, and an edge $e^*$ for each edge $e$ in $G$ (including boundary edges) joining the vertices dual to the faces that~$e$ separates.  For each boundary cycle $\delta$ of~$G$, we refer to the corresponding vertex $\delta^*$ of $G^*$ as a \EMPH{dual boundary vertex}.  The dual graph $G^*$ has a natural cellular embedding in the surface~\EMPH{$\Sigma^\bullet$} obtained from~$\Sigma$ by gluing a disk to each boundary cycle; each face of this embedding corresponds to a vertex of $G$.  See Figure \ref{fig:duality}.  (Duality can be extended to directed graphs~\cite{cen-hfcc-12}, but our results do not require this extension.)

\begin{figure}[htb]
\centering
\includegraphics[scale=0.45]{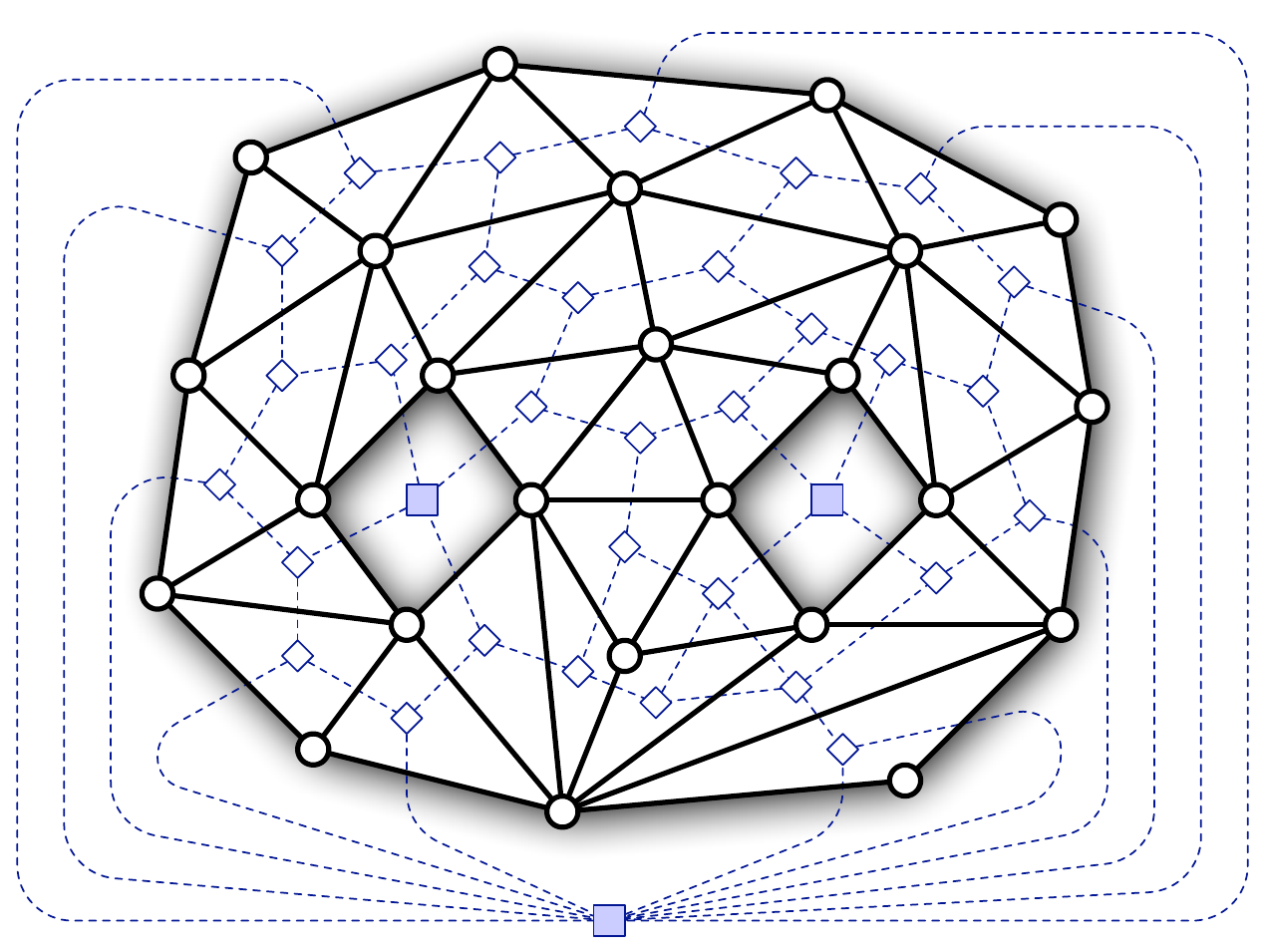}
\caption{A cellularly embedded graph $G$ (solid lines) on a pair of pants (the surface of genus 0 with 3 boundaries), and its dual graph $G^*$ (dashed lines).  Dual boundary vertices are indicated by squares.}
\label{fig:duality}
\end{figure}
 
For any subgraph $F = (U,D)$ of $G = (V,E)$, we write \EMPH{$G\setminus F$} to denote the edge-complement $(V, {E\setminus D})$.  Also, when the graph $G$ is fixed, we abuse notation by writing~$F^*$ to denote the subgraph of $G^*$ corresponding to a subgraph $F$ of~$G$; each edge in $F^*$ is the dual of a unique edge in~$F$.  In particular, we have the identity $(G\setminus F)^* = G^* \setminus F^*$.
Further, we may sometimes use~$D$ to refer to an edge set or the subgraph~$F = (V, D)$,
but it should be clear which we mean from context.


\subsection{Homotopy and homology}
\label{SS:homology}

Two paths~$p$ and~$q$ in $\Sigma$ are \EMPH{homotopic} if one can be continuously deformed into the other without changing their endpoints.
More formally, a \EMPH{homotopy} between~$p$ and~$q$ is a continuous map $h\colon {[0,1]\times [0,1] \to \Sigma}$ such that $h(0,\cdot) = p$, $h(1,\cdot) = q$, $h(\cdot, 0)=p(0)=q(0)$, and $h(\cdot,1)=p(1)=q(1)$.  Homotopy defines an equivalence relation over the set of paths with any fixed pair of endpoints.

Similarly, two cycles $\alpha$ and $\beta$ in $\Sigma$ are \EMPH{freely homotopic} if one can be
continuously deformed into the other.  More formally, a free homotopy between~$\alpha$ and~$\beta$
is a continuous map $h\colon {[0,1]\times S^1 \to \Sigma}$ such that $h(0,\cdot) = \alpha$ and
$h(1,\cdot) = \beta$.  Free homotopy defines an equivalence relation over the set of cycles
in~$\Sigma$.  We omit the word ``free'' when it is clear from context.

%
%
A cycle is \EMPH{contractible} if it is homotopic to a constant map.
Given a weight function on the edges of~$G$, we say a path or cycle is \EMPH{tight} if it has minimum total weight (counting edges with multiplicity) for its homotopy class.

Homology is a coarser equivalence relation than homotopy, with nicer
algebraic properties.  Like several earlier papers \cite{cf-qhc2-07,
cf-qhc-08, dls-chtl-07, dlsc-cgaht-08,e-sncds-11,f-sntcd-13}, we will consider only
one-dimensional cellular homology with coefficients in the finite
field $\Z_2$; this restriction allows us to radically simplify our
definitions.
Fix a cellular embedding of an undirected graph $G$ on a surface with genus $g$ and $b$ boundaries.  A \EMPH{boundary subgraph} is the boundary of the union of a subset of faces of $G$; for example, every separating cycle is a boundary subgraph.
Two even subgraphs are \EMPH{homologous}, or in the same \EMPH{homology class}, if their symmetric difference is a boundary subgraph.
Boundary subgraphs are also called \EMPH{null-homologous}.  Any two homotopic cycles are homologous, but homologous cycles are not necessarily homotopic; see Figure \ref{fig:homology}.  Moreover, the homology class of a cycle can contain even subgraphs that are not cycles; see Figure \ref{fig:homology2}.
Homology classes define a vector space $\Z_2^\beta$, called the first homology group, where $\beta = 2g$ if the surface has no boundary and $\beta = 2g+b-1$ otherwise.
The rank $\beta$ of the first homology group is called the first \emph{Betti number} of the surface.

\begin{figure}[htb]
\centering
\includegraphics[height=1in]{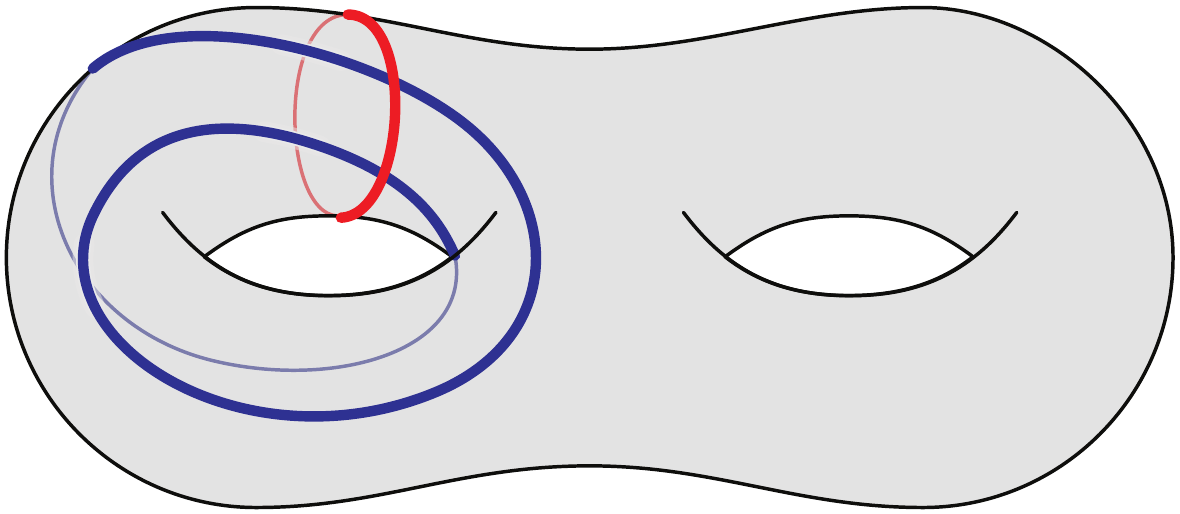}\qquad
\includegraphics[height=1in]{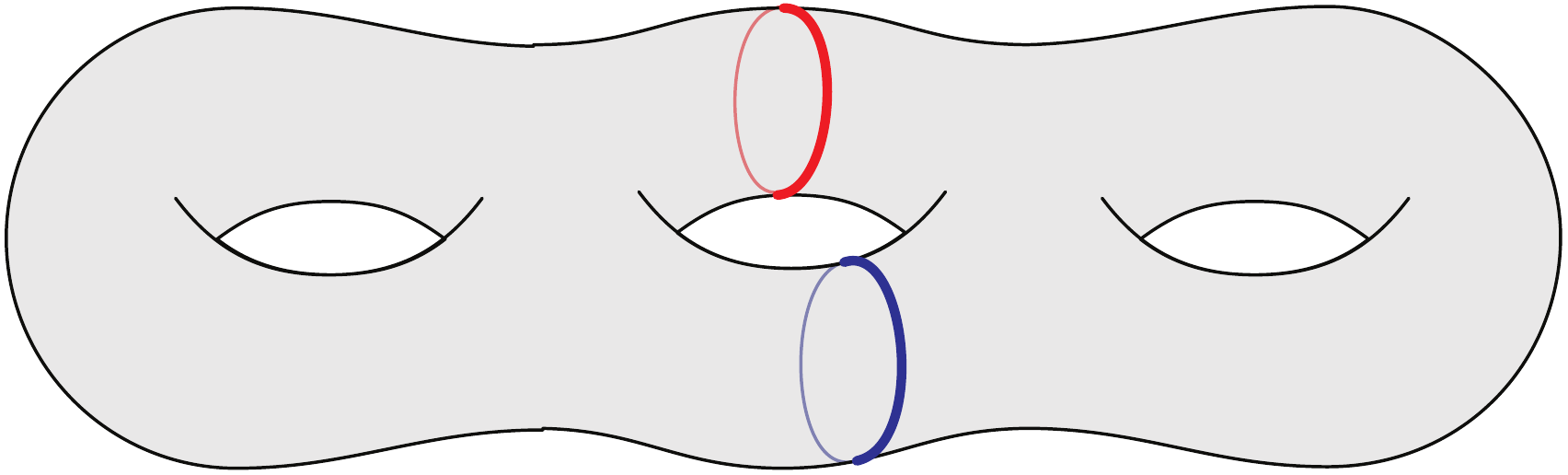}
\caption{Homologous pairs of cycles that are not homotopic.  (Lighter portions of the curves are on the back side of the surface.)}
\label{fig:homology}
\end{figure}

\begin{figure}[htb]
\centering
\includegraphics[height=1in]{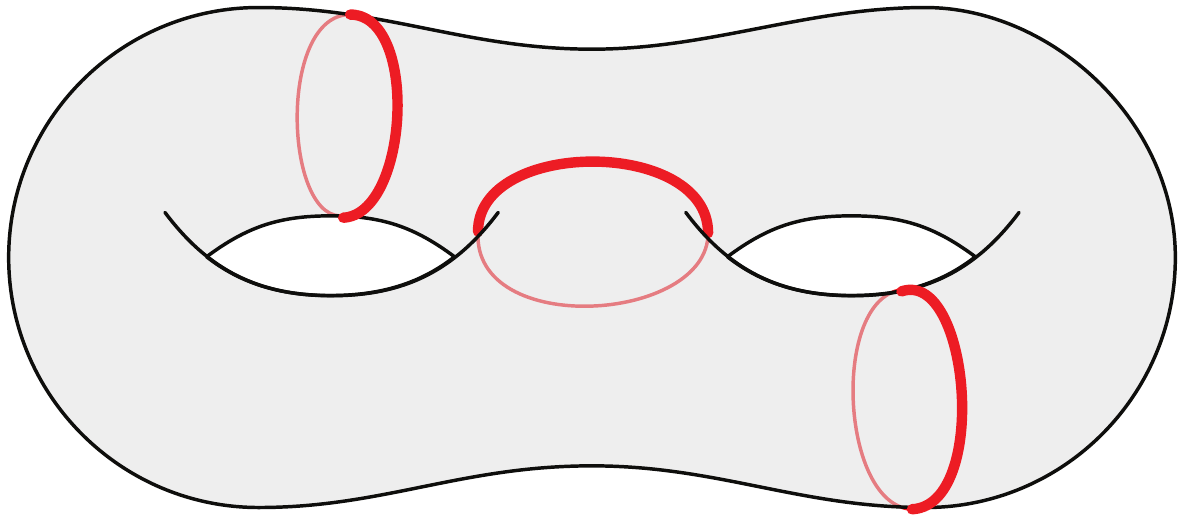}
\caption{Each cycle is homologous to the union of the other two.}
\label{fig:homology2}
\end{figure}

\subsection{Duality between cuts and even subgraphs}
\label{SS:duality}

A crucial component of our minimum $(s,t)$-cut algorithms is an equivalence between $(s,t)$-cuts and even subgraphs of the \emph{dual graph} contained in a particular homology class.  This equivalence was first observed in planar graphs by Whitney \cite{w-pg-33} and was later used to compute minimum cuts in planar graphs by Itai and Shiloach \cite{is-mfpn-79}.
We formalize the same equivalence on surface graphs in the following lemma:

\begin{lemma}
\label{lem:cut-duality}
Let $G$ be an edge-weighted graph embedded on a surface $\Sigma$ without boundary, and let $s$ and~$t$ be vertices of~$G$.  If $\Cut$ is an $(s,t)$-cut in~$G$, then $\Cut^*$ is an even subgraph of $G^*$ homologous with the boundary of $s^*$ in the surface $\Sigma\setminus(s^*\cup t^*)$.  In particular, if $X$ is a minimum-weight $(s,t)$-cut in~$G$, then~$\Cut^*$ is the minimum-weight even subgraph of $G^*$ homologous with the boundary of $s^*$ in $\Sigma\setminus(s^*\cup t^*)$.
\end{lemma}

\begin{proof}
Let $\partial s^*$ denote the boundary of $s^*$, and let $\Sigma'$ denote the surface $\Sigma\setminus {(s^*\cup t^*)}$.

Let $\Cut$ be an arbitrary $(s,t)$-cut in $G$.  This cut partitions the vertices of $G$ into two disjoint subsets~$S$ and $T$, respectively containing vertices $s$ and $t$.  Thus, the dual subgraph $\Cut^*$ partitions the faces of $G^*$ into two disjoint subsets, $S^*$ and $T^*$, respectively containing faces $s^*$ and $t^*$.  In particular, $\Cut^*$ is the boundary of the union of the faces in $S^*$, which implies that $\Cut^*$ is null-homologous in $\Sigma$.  The subgraph $\Cut^* \oplus \partial s^*$ is the boundary of the union of $S^* \setminus\set{s^*}$, which is a subset of the faces of $\Sigma'$.  Thus, $\Cut^*\oplus \partial s^*$ is null-homologous in $\Sigma'$.  We conclude that $\Cut^*$ and  $\partial s^*$ are homologous in $\Sigma'$.

Conversely, let $\Cut^*$ be an arbitrary even subgraph of $G^*$ homologous to $\partial s^*$ in $\Sigma'$.  The subgraph $\Cut^*\oplus \partial s^*$ is null-homologous in $\Sigma'$.  This immediately implies that $X^*$ is null-homologous in $\Sigma$; moreover, faces $s^*$ and $t^*$ are on opposite sides of $X^*$.  Any path from $s$ to $t$ in the original graph $G$ must traverse at least one edge of $\Cut$.  We conclude that $\Cut$ is an $(s,t)$-cut.
\end{proof}

\section{Characterizing Homology}
\label{S:tree-cotree}

Throughout the paper, we fix an \EMPH{undirected} graph $G=(V,E)$, a non-negative weight function $w\colon E\to \Real$, and a cellular embedding of $G$ on a surface $\Sigma$ of genus $g$ with $b$ boundary cycles.
Except where explicitly indicated otherwise, we assume without loss of generality  that the underlying surface $\Sigma$ has non-empty boundary; otherwise, we can remove an arbitrary face of $G$ from~$\Sigma$ without affecting its homology at all.  Let $\delta_1, \dots, \delta_b$ denote the boundary cycles of $\Sigma$, and let $\beta = 2g+b-1$ denote the the first Betti number of $\Sigma$.

In this section, we describe two standard methods for preprocessing a combinatorial surface with boundary in~$O(\beta n)$ time, so that the $\Z_2$-homology class of any even subgraph $H$ can be computed in $O(\beta)$ time per edge.  These are both straightforward generalizations of standard methods for measuring homology in surfaces \emph{without} boundary based on tree-cotree decompositions \cite{ew-gohhg-05, ccelw-scsih-08, e-dgteg-03}.
We give these full details here completeness, and because as far as we are aware, no detailed description appears
elsewhere in the literature for the first method.  We note that a preliminary version of the current
work~\cite{en-mcsnc-11} was the first detailed description of the second method;
see also Chambers~\etal~\cite{bcfn-mchbs-17} for an alternative description of the second method.
All results in this section extend without modification to nonorientable surfaces.

Both methods characterize the homology class of any even subgraph $H$ using a vector of~$\beta$ bits.
The vectors are computed using a one of two natural generalizations of tree-cotree decompositions~\cite{e-dgteg-03} to surfaces with boundary.
In the first method, the vector is based on the crossings between a cycle decomposition of~$H$ and a set of~$\beta$ primal arcs.
By carefully selecting these arcs, we can bound the number of times any~$\Z_2$-minimal even subgraph can cross any of these arcs; this bound is necessary for the algorithm given in Section~\ref{sec:crossing}.
In the second method, the vector is based on the crossings between~$H$ and a set of~$\beta$ \emph{dual} arcs.
The second method is somewhat easier to describe and implement than the first, so we use the second method in the algorithm given in Section~\ref{sec:homcover}.


\subsection{Crossing parity vectors via forest-cotree decompositions}
\label{sec:characterizing_crossings}

The first method begins by computing a set~$A$ of~$\beta$ arcs, each of which is the concatenation of two shortest paths (possibly meeting in the interior of an edge), such that the surface $\Sigma\setminus A$ is a topological disk.
Following previous papers \cite{ccelw-scsih-08, ce-tnpcs-10, c-scgsp-10}, we construct a \EMPH{greedy system of arcs}, using a variant of Erickson and Whittlesey's algorithm to construct optimal systems of loops \cite{ew-gohhg-05}.  Our algorithm uses a natural generalization of tree-cotree decompositions~\cite{e-dgteg-03} to surfaces with boundary.

A \EMPH{forest-cotree decomposition} of $G$ is any partition $(\partial\! G, F, L, C)$ of the edges of $G$ into four edge-disjoint subgraphs with the following properties:
\begin{itemize}\itemsep0pt
\item $\partial\! G$ is the set of all boundary edges of $G$.
\item $F$ is a spanning forest of $G$, that is, an acyclic subgraph of $G$ that contains every vertex.
\item Each component of $F$ contains a single boundary vertex.
\item $C^*$ is a spanning tree of $G^*\setminus (\partial G)^*$, that is, a subtree of $G^*$ that contains every vertex \emph{except} the dual boundary vertices $\delta_i^*$.
\item Finally, $L$ is the set of leftover edges $E \setminus (\partial\!G \cup F \cup C^*)$
\end{itemize}

\noindent
Euler's formula implies that there are exactly $\beta$ leftover edges in $L$; arbitrarily label these edges $e_1, e_2, \dots, e_\beta$.  For each edge $e_i\in L$, the subgraph~$F\cup \set{e_i}$ contains a single nontrivial arc $a_i$, which is either a simple path between distinct boundary cycles, or a nontrivial loop from a boundary cycle back to itself; in the second case, $a_i$ may traverse some edges of $F$ twice.  Slicing along the arcs $a_1, \dots, a_\beta$ transforms~$\Sigma$ into a topological disk. See Figure \ref{fig:forest-cotree}.

\begin{figure}[htb]
\centering\footnotesize\sf
\begin{tabular}{c}
\includegraphics[scale=0.45]{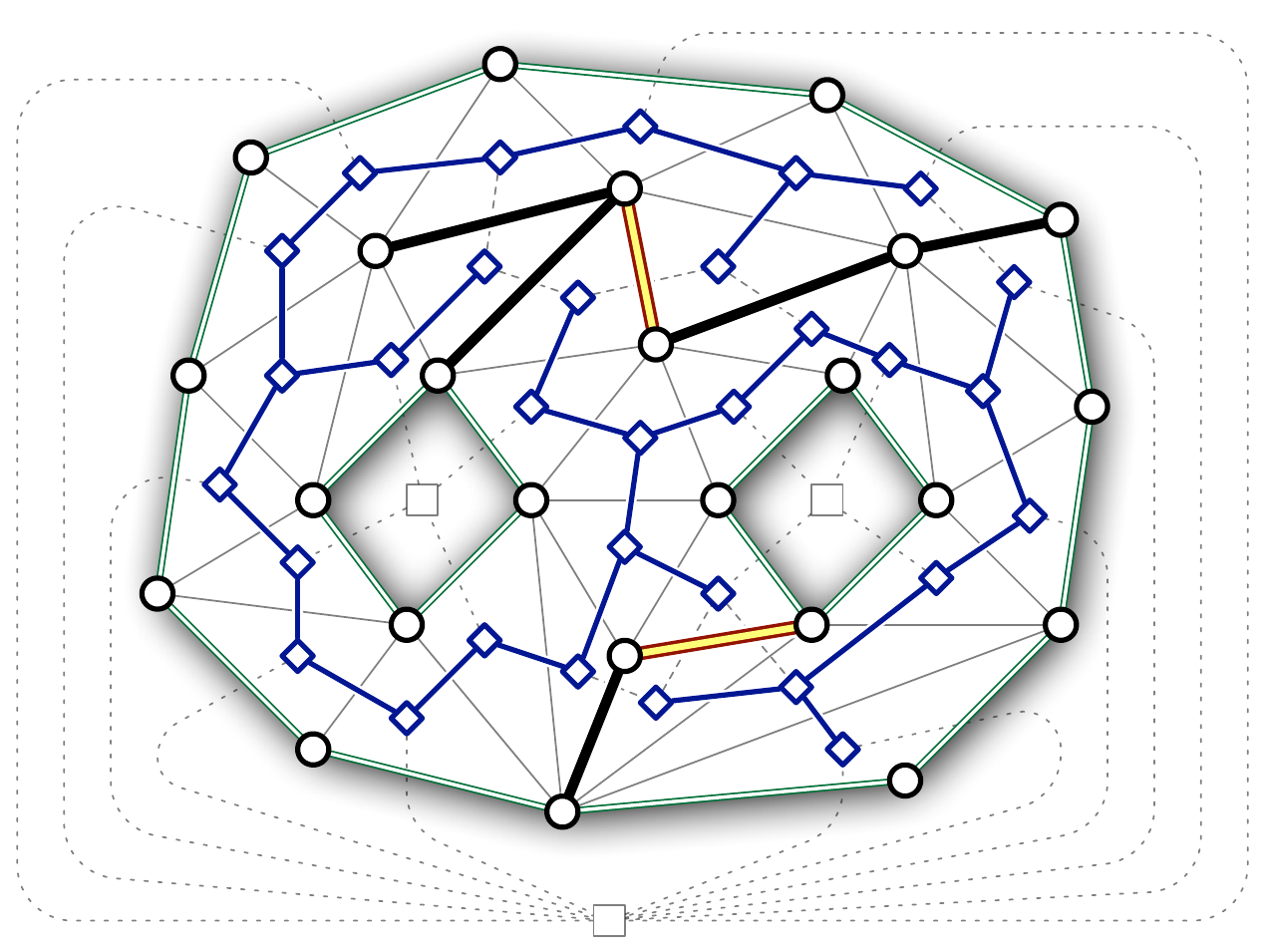} \qquad
\includegraphics[scale=0.45]{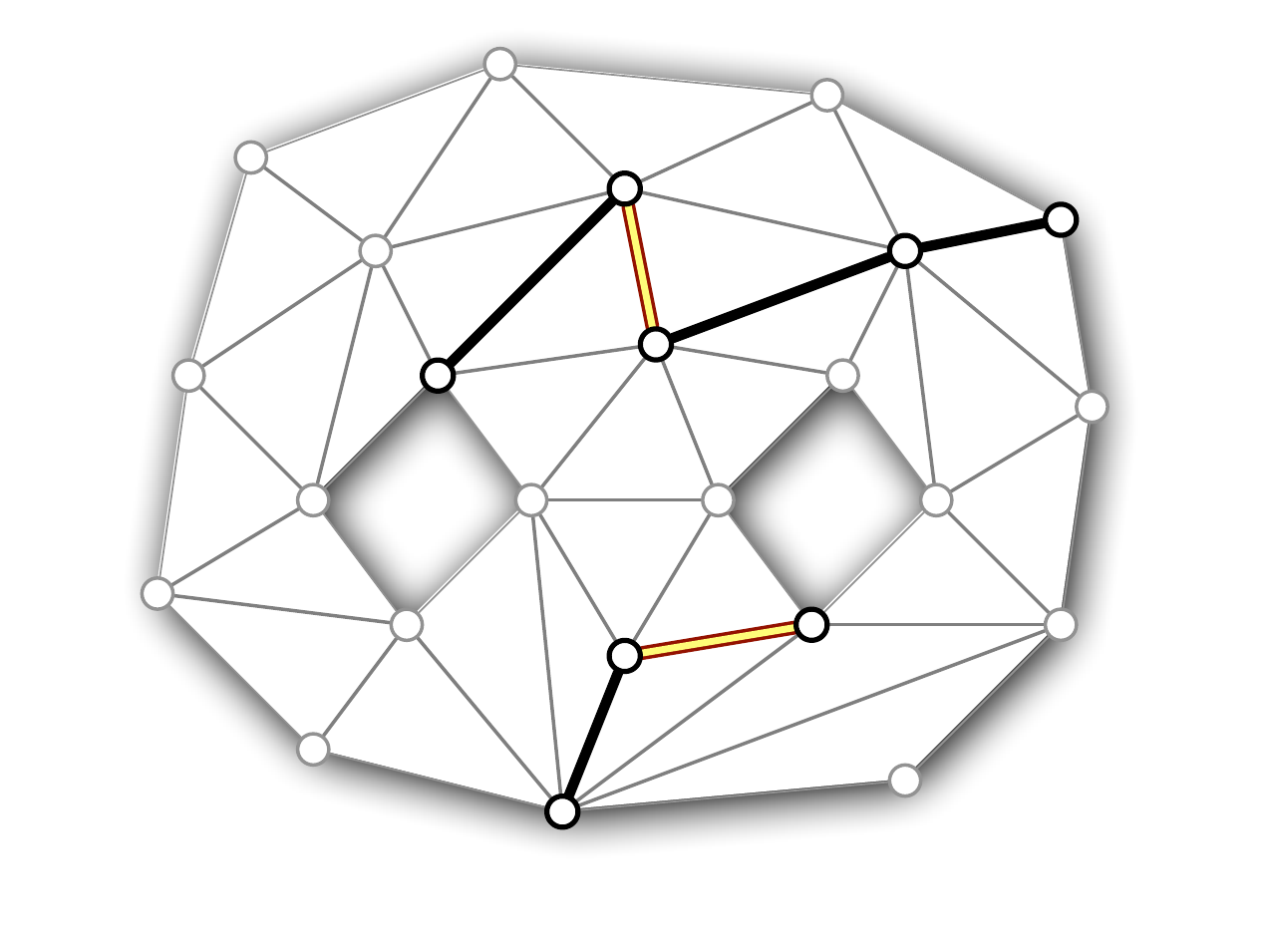}
\end{tabular}
\caption{Left: A forest-cotree decomposition of the graph in Figure \ref{fig:duality}; thick doubled lines indicate edges in $L$.
Right: The resulting system of arcs.  Compare with Figure \ref{fig:tree-coforest}.}
\label{fig:forest-cotree}
\end{figure}

For any cycle $\gamma$ and any index $i$, let $x_i(\gamma)$ denote the number of times $\gamma$ crosses the arc~$a_i$.  The \EMPH{crossing vector} $x(\gamma)$ of $\gamma$ is the vector $(x_1(\gamma), \dots, x_\beta(\gamma))$.  The crossing vector of a set of cycles is the sum of the crossing vectors of its elements.

Extending the notion of crossing vectors to even subgraphs is rather subtle, because we cannot consistently define when a path crosses an even subgraph $H$.  Instead, we consider crossings between a path and the cycles in an arbitrary cycle decomposition of the even subgraph $H$.  Different cycle decompositions may yield different numbers of crossings, so even subgraphs do not have well-defined crossing vectors; however, the \emph{parity} of the crossing number is independent of the cycle decomposition.  The \EMPH{crossing parity vector} of any even subgraph $H$ is the bit vector $\bar{x}(H) = (\bar{x}_1, \dots, \bar{x}_\beta)$, where $\bar{x}_i = 1$ if the arc $a_i$ crosses (any cycle decomposition of) $H$ an odd number of times, and $\bar{x}_i = 0$ otherwise.

\begin{lemma}
Two even subgraphs are $\Z_2$-homologous if and only if their crossing parity vectors (with respect to the same system of arcs) are equal.
\end{lemma}

\begin{proof}
Every boundary subgraph is the symmetric difference of facial cycles.  Any non-contractible loop or arc crosses any facial cycle an even number of times; thus, the crossing parity vector of any facial cycle is the zero vector.  Every pair of even subgraphs $H$ and $H'$ satisfies the identity $x(H\oplus H') = x(H) \oplus x(H')$.  Thus, the crossing parity vector of any boundary subgraph is the zero vector.
\end{proof}

\begin{lemma}
We can compute the crossing parity vector of any even subgraph, with respect to any fixed system of
arcs, in $O(\beta)$ time per edge of the subgraph.
\end{lemma}

\begin{proof}
We can compute a cycle decomposition $\gamma_1, \dots, \gamma_r$ of $H$ in $O(1)$ time per edge, by following the proof of Lemma~\ref{lem:decomposition}.
We can compute the number of crossings between any cycle $\gamma_i$ and any arc $a_j$ in time proportional to the number of edges in $\gamma_i$.
\end{proof}

We can easily construct an \emph{arbitrary} forest-cotree decomposition, and thus an arbitrary system of arcs, in $O(n)$ time using whatever-first search, but our algorithms require a decomposition with a particular forest $F$ and a particular dual spanning tree $C^*$.  Let $G/\partial G$ denote the graph obtained from $G$ by \emph{contracting} the entire subgraph $\partial G$---both vertices and edges---to a single vertex~$x$.  Using the algorithm of Henzinger \etal~\cite{hkrs-fspap-97}, we compute the single-source shortest-path tree $T$ in $G/\partial G$ rooted at~$x$ in $O(n)$ time.\footnote{This running time requires that $g = O(n^{1-\e})$ for some constant $\e>0$.  However, we can safely assume $g = o(\log n)$, since otherwise our minimum-cut algorithms are slower than textbook algorithms for arbitrary graphs.}  Let $F$ be the subgraph of $G$ corresponding to~$T$.  Each component of $F$ is a tree of shortest paths from a boundary vertex to a subset of the non-boundary vertices of $G$.  

Now for each edge $e$ that is \emph{not} in the forest $F$ or the boundary subgraph $\partial G$, let $\ell(e)$ denote the length of the unique arc in the subgraph $F \cup\set{e}$.  We can easily compute $\ell(e)$ for each non-forest edge $e$ in $O(n)$ time.  Finally, let $C^*$ denote the maximum spanning tree of $G^* \setminus (F\cup \partial G)^*$ with respect to the arc lengths $\ell(e)$.

Finally, for each edge $e_i\in L$, let $\sigma_i$ and $\tau_i$ denote the unique directed paths in $F$ from the boundary of~$G$ to the endpoints of $e_i$, and let $S := \set{\sigma\!_1, \dots, \sigma\!_\beta, \allowbreak \tau_1, \dots, \tau_\beta}$.  By construction of $F$, every element of~$S$ is a (possibly empty) shortest directed path.  Moreover, because $a_i = \sigma_i \cdot e_i \cdot \rev(\tau_i)$ for each index~$i$, every non-null-homologous cycle in $G$ must intersect at least one path in $S$.  We can easily compute each path in $S$ in $O(n)$ time.  The final greedy system of arcs is the set $A := \set{a_1, a_2, \dots, a_\beta}$.

Exchange arguments by Erickson and Whittlesey \cite{ew-gohhg-05} and Colin de Verdière \cite{c-scgsp-10} both imply that every arc in the greedy system is tight, and moreover that the greedy system of arcs has minimum total length among all systems of arcs.\footnote{Specifically, Colin de Verdière’s argument implies that the greedy system of arcs is a minimum-length basis in $G$ for the first relative homology group $H_1(\Sigma, \partial\Sigma)$ \cite[Section 3]{c-scgsp-10}.  Thus, each arc in the greedy system is as short as possible in its \emph{relative homology} class.}

\subsection{Homology signatures via tree-coforest decompositions}
\label{sec:characterizing_signatures}

Our second method associates a vector of $\beta$ bits with each edge $e$, called the \EMPH{signature} of~$e$; the homology class of any even subgraph is characterized by the bit-wise exclusive-or of the signatures of its edges.

Again, our construction is based on one of two natural generalizations of tree-cotree
decompositions~\cite{e-dgteg-03} to surfaces with boundary; the other generalization is used for
computing crossing parity vectors as described above.
We define a \EMPH{tree-coforest decomposition} of $G$ to be any partition $(T, L, F)$ of the edges of $G$ into three edge-disjoint subgraphs with the following properties:
\begin{itemize}\itemsep0pt
\item $T$ is a spanning tree of $G$.
\item $F^*$ is a spanning \emph{forest} of $G^*$, that is, an acyclic subgraph that contains every vertex.
\item Each component of $F^*$ contains a single dual boundary vertex $\delta_i^*$.
\item Finally, $L$ is the set of leftover edges $E \setminus (T\cup F)$.
\end{itemize}
Euler's formula implies that there are exactly~$\beta$ edges in $L$; arbitrarily index these edges $e_1, \dots, e_\beta$.  For each edge $e_i\in L$, adding the corresponding dual edge $e_i^*$ to $F^*$ creates a new dual path $\dualarc_i$, which is either a simple path between distinct boundary vertices, or a nontrivial loop from a boundary vertex back to itself; in the second case, $\dualarc_i$ may traverse some edges of $F^*$ twice.  We can treat each path $\dualarc_i$ as a simple arc in the \emph{abstract} surface $\Sigma$; slicing along these $\beta$ arcs transforms $\Sigma$ into a topological disk.
See Figure~\ref{fig:tree-coforest}.  We call the set $\set{\dualarc_1, \dualarc_2, \dots, \dualarc_\beta}$ a \EMPH{system of dual arcs}.

\begin{figure}[htb]
\centering\footnotesize\sf
\includegraphics[scale=0.45]{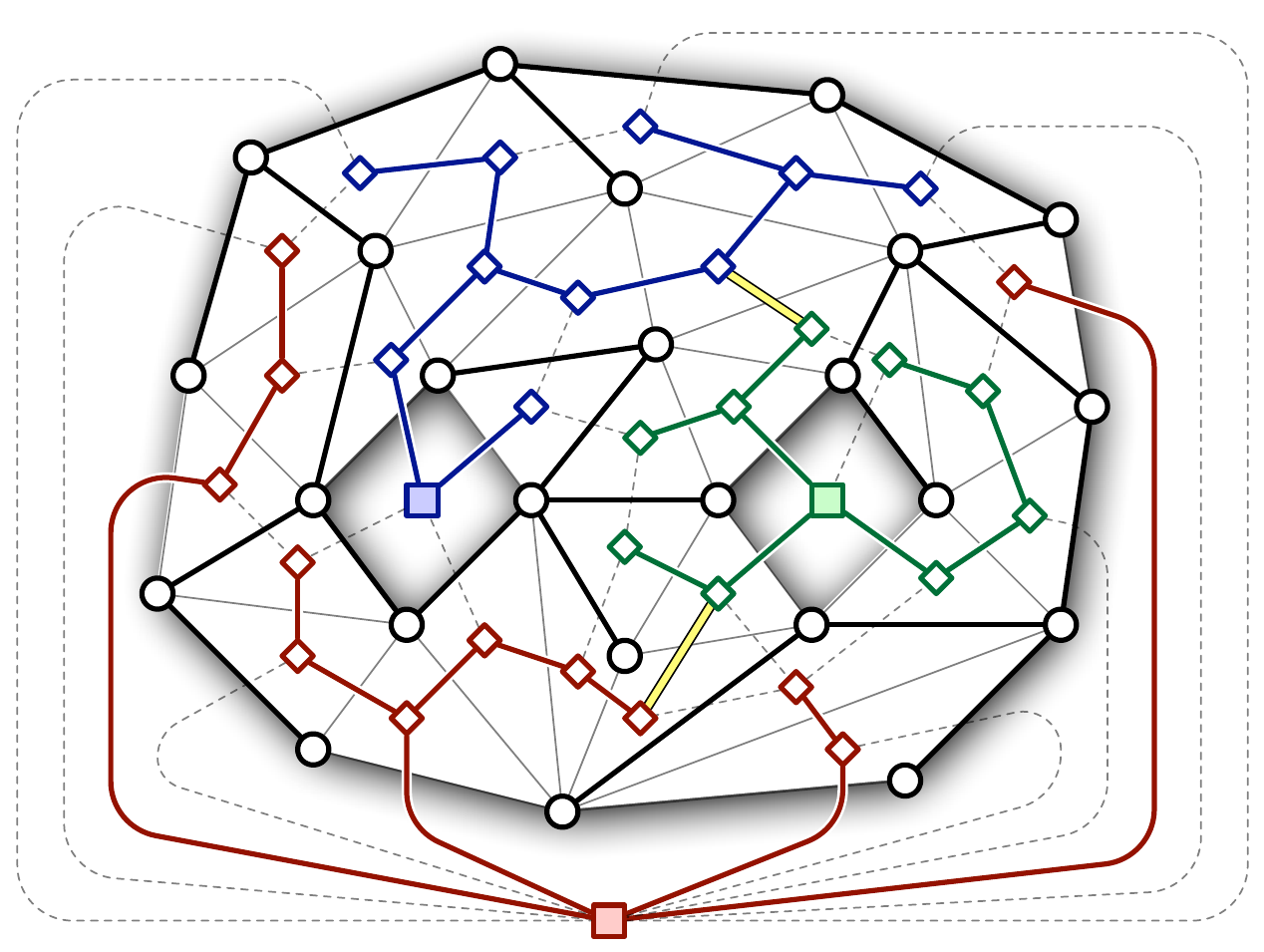} \qquad
\includegraphics[scale=0.45]{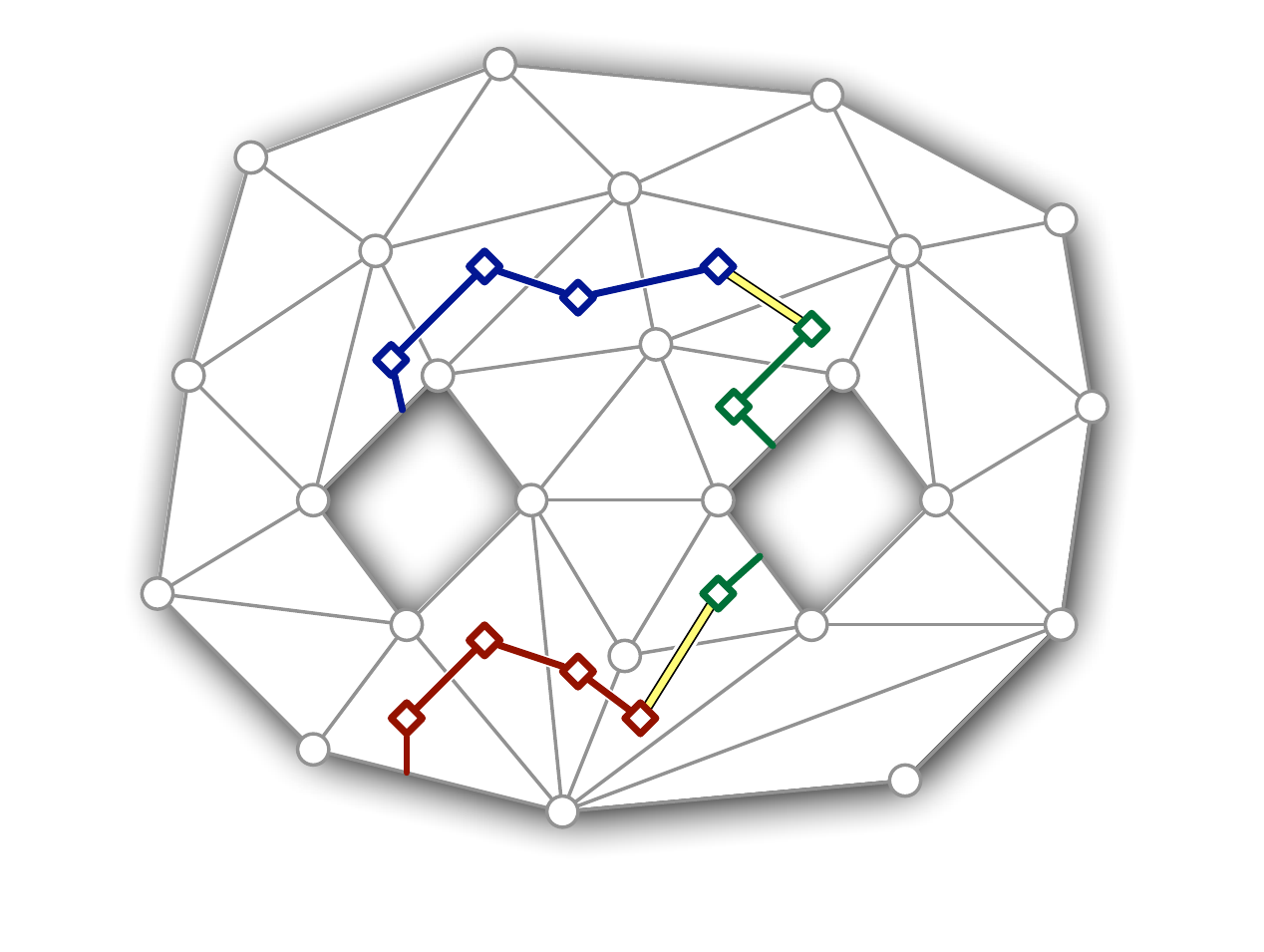}
\caption{Left: A tree-coforest decomposition of the graph in Figure~\ref{fig:duality}; doubled lines indicate edges in $L$.
Right: The resulting system of dual arcs.  Compare with Figure \ref{fig:forest-cotree}.}
\label{fig:tree-coforest}
\end{figure}

Finally, for each edge $e$ in $G$, we define its signature \EMPH{$[e]$} to be the $\beta$-bit vector whose $i$th bit is equal to~$1$ if and only if $e$ crosses $\dualarc_i$ (that is, if $\dualarc_i$ traverses the dual edge $e^*$) an odd number of times. The signature $[H]$ of an even subgraph $H$ is the bitwise exclusive-or of the signatures of its edges.  Similarly, the signature~$[\cycle]$ of a cycle $\cycle$ is the bitwise exclusive-or of the signatures of the edges that $\cycle$ traverses an odd number of times.

Let \EMPH{$h \oplus h'$} denote the bitwise exclusive-or of two homology signatures $h$ and $h'$, or equivalently, their sum as elements of the homology group~$(\Z_2)^\beta$.  The identities $[H \oplus H'] = [H] \oplus [H']$ and $[\cycle\cdot\cycle'] = [\cycle] \oplus [\cycle']$ follow directly from the definitions.

\begin{lemma}
\label{lem:sign}
We can preprocess $G$ in $O(\beta n)$ time, so that the signature $[\cycle]$ of any cycle can be computed in $O(\beta)$ time per edge.
\end{lemma}

\begin{proof}
A tree-coforest decomposition can be computed in $O(n)$ time as follows.  First construct a graph~$G'$ by identifying all the dual boundary vertices in $G^*$ to a single vertex.  Compute a spanning tree of $G'$ by whatever-first search; the edges of this spanning tree define an appropriate dual spanning forest~$F^*$.  Construct the subgraph $G\setminus F$ and compute a spanning tree $T$ via whatever-first search.  Finally, let $L = G\setminus (T\cup F)$.  With the decomposition in hand, it is straightforward to compute each path $\dualarc_i$ in $O(n)$ time, and then compute each edge signature in $O(\beta)$ time.
\end{proof}

\begin{lemma}
An even subgraph $H$ of $G$ is null-homologous in $\Sigma$ if and only if $[H] = 0$.
\end{lemma}

\begin{proof}
Let $H$ be a boundary subgraph of $G$.  Then by definition, $H$ is the boundary of the union of a subset~$Y$ of faces of $G$.  The boundary of any face $f$ is contractible in $\Sigma$ and therefore has signature $0$.  It follows immediately that $[H] = [\bigoplus_{f\in Y} \partial f] = \bigoplus_{f\in Y} [\partial f] = 0$.

Conversely, suppose $H$ crosses each arc $\dualarc_i$ an even number of times, so $[H]=0$.  Let $x$ and $y$ be two intersection points between $H$ and some arc $\dualarc_i$, and let $\dualarc_i[x,y]$ be the subpath of $\dualarc_i$ between those two points.  Replacing tiny segments of $H$ through~$x$ and~$y$ with two copies of $\dualarc_i[x,y]$ does not change the homology class of $H$, but does reduce the number of intersection points between $H$ and $\dualarc_i$.  It follows by induction that~$H$ is homologous to another even graph $H'$ that does not intersect any path~$\dualarc_i$ at all.  This even graph lies entirely within the disk $\Sigma\setminus \bigcup_i\dualarc_i$, and is therefore null-homologous.
\end{proof}


The following corollaries are now immediate.

\begin{corollary}
Two even subgraphs $H$ and $H'$ of $G$ are $\Z_2$-homologous in $\Sigma$ if and only if $[H] = [H']$.
\end{corollary}

\begin{corollary}
Two cycles $\cycle$ and $\cycle'$ in $G$ are $\Z_2$-homologous in $\Sigma$ if and only if $[\cycle] = [\cycle']$.
\end{corollary}

\section{Crossing Bounds and Triangulations}
\label{sec:crossing}

In this section, we describe an
algorithm to compute the minimum-weight even subgraph homologous with any specified even subgraph 
$H$ in $(g+b)^{O(g+b)}n\log \log n$ time.  In fact, our algorithm can be
modified easily to compute a minimum-weight representative in
\emph{every} homology class in the same asymptotic running time;
there are exactly $2^{2g+b-1}$ such classes.
Lemma~\ref{lem:cut-duality} implies our algorithm can be used to find a minimum $(s,t)$-cut in~$G^*$ in the same amount of time.

Our algorithm closely resembles the algorithm of Chambers \etal~\cite{ccelw-scsih-08} for computing a shortest splitting cycle; in fact, our algorithm is somewhat simpler.  Our algorithm is based on the key observation (Lemma \ref{lem:crossing})  that the shortest even subgraph in any homology class crosses any shortest path at most $O(g+b)$ times.  The first stage of our algorithm cuts the underlying combinatorial surface into a topological disk by a greedy system of arcs, as described in Section~\ref{sec:characterizing_crossings}.  Next, we enumerate all possible ways for an even subgraph to intersect each of the greedy arcs at most $O(g+b)$ times; we quickly discard any crossing pattern that does not correspond to an even subgraph in the desired homology class.  Each crossing pattern is realized by several (free) \emph{homotopy} classes of sets of non-crossing cycles; we show how to enumerate these homotopy classes in Section~\ref{SS:homotopy-triangulation}.  Then within each homotopy class, we find a minimum-length set of non-crossing cycles with each crossing pattern, essentially by reducing to a \emph{planar} instance of the minimum-cut problem.  The union of those cycles is an even subgraph in the desired homology class; we return the lightest such subgraph as our output.

\subsection{Crossing bound}
\label{SS:homotopy-crossing}

Our main technical lemma for this section establishes an upper bound on the number of crossings
between members of a greedy system of arcs and some minimum-weight even subgraph in any homology class.  Crossing-number arguments were first used by Cabello and Mohar \cite{cm-fsnsn-07} to develop the first subquadratic algorithms for shortest non-contractible and non-separating cycles in undirected surface embedded graphs; their arguments are the foundation of all later improvements of their algorithm \cite{c-mdpg-06, k-csnco-06, cce-msspe-13}.  Our proof is quite similar to the argument of Chambers \etal~\cite{ccelw-scsih-08} that the shortest \emph{splitting} cycle crosses any shortest path $O(g+b)$ times.  However, our new proof is simpler, because the structure we seek is a true subgraph, which need not be connected, rather than a single (weakly) simple closed walk.

As mentioned in Section~\ref{sec:characterizing_crossings}, we cannot consistently define when a shortest path crosses an even subgraph.  Instead, we consider the total number of crossings between a shortest path and the cycles in an arbitrary cycle decomposition.

\begin{lemma}
\label{lem:crossing}
Let $G$ be an edge-weighted graph embedded on a surface with genus $g$ and $b$ boundary components.
Let $A = \Set{a_1, a_2, \dots, a_\beta}$ be a greedy system of arcs.
Let $H$ be a subgraph of $G$.
There is a $\Z_2$-minimal even subgraph $H'$ homologous to $H$ such that for any cycle decomposition
$\gamma_1, \gamma_2, \dots, \gamma_r$ of $H'$, the total number of crossings between any arc $a_i$
and the cycles $\gamma_1, \gamma_2, \dots, \gamma_r$ is at most $12g+4b-5$.
\end{lemma}

Our proof begins by conceptually modifying the edge weights in $G$ in a manner reminiscent to the
way many earlier papers \cite{cce-msspe-13, benw-amcnt-16, mnnw-mdpgo-18, bsw-msopg-15, e-pspmf-10,
eh-ocsd-04} enforce edge weight \emph{genericity}, and in particular, the \emph{uniqueness} shortest
paths between any two vertices of a graph.
This assumption can be enforced \emph{with high probability} by adding random infinitesimal weights
to each edge \cite{mvv-memi-87}.
Cabello \etal~\cite{cce-msspe-13} describe an efficient
implementation of lexicographic perturbation \cite{c-odlp-52,dow-gsmml-55,hm-apmcb-94} that
increases the worst-case running times of algorithms by a factor of $O(\log n)$.
More recently, Erickson \etal \cite{efl-hmcpf-18} described a deterministic perturbation scheme for
directed graphs based on integer homology that increases worst-case running times by a factor of
$O(g)$.
Unfortunately, this latter scheme cannot be used as a black box in algorithms such as ours that rely
on edges being undirected.


\begin{proof}
For the sake of argument, we slightly modify the edge weights of $G$ so that $Z_2$-minimal even
subgraphs under the new weights are also $Z_2$-minimal under the original edge weights but the
crossing bound is guaranteed to hold.
Recall, every arc in $A$ consists of two shortest paths from a forest $F$ of shortest paths.
Our modification simply adds identical infinitesimal weights $\eps$ to every edge \emph{outside} of
$F$.
We claim that for any shortest path $\sigma$ contributing to $F$, for any pair of vertices $u$ and $v$
in $\sigma$, that $\sigma(u, v)$ is the unique shortest path from $u$ to $v$.
Indeed, any path from $u$ to $v$ that uses an edge outside $F$ must weigh at least $\eps$ more than
$\sigma(u, v)$, and $F$ is a forest so there are no other paths from $u$ to $v$ using only edges of
$F$.

Let $H'$ be is an arbitrary $\Z_2$-minimal even subgraph homologous to $H$ under the new edge weights.
Let $\gamma_1, \gamma_2, \dots, \gamma_r$ be an arbitrary cycle decomposition of $H'$.
We now argue that for any shortest path $\sigma = \sigma(u, v)$ contributing to $F$, the total
number of crossings between $\sigma$ and the cycles $\gamma_1, \gamma_2, \dots, \gamma_r$ is at most
$6g+2b-3$.
The lemma immediately follows from the construction of $A$.
Without loss of generality, we can assume that $\sigma$ crosses each cycle $\gamma_i$ at least once.  For each $i$, let~$x_i$ denote the number of times $\sigma$ and $\gamma_i$ cross, and let $x = x_1 + x_2 + \cdots + x_r$.  We need to prove that $x\le 6g+2b-3$.


Consider the graph $G/\sigma$ obtained from $G$ by contracting the path $\sigma$ to a single vertex $uv$.  This graph inherits a cellular embedding on $\Sigma$ from the cellular embedding of $G$.  Each cycle $\gamma_i$ is contracted to the union of~$x_i$ weakly simple non-crossing loops in $G/\sigma$ with basepoint $uv$.  Altogether, we obtain $x$ loops, which we denote $\ell_1, \ell_2, \dots, \ell_x$.  We claim that these $x$ loops lie in distinct nontrivial homotopy classes.

Suppose some loop $\ell_i$ is contractible.  This loop is the contraction of a path $\pi_i$ in $G$
whose endpoints~$u_i$ and $v_i$ lie in $\sigma$.  The cycle $\delta = \pi_i \cdot \sigma(v_i,u_i)$
is also contractible.  Thus, the even subgraph $H\oplus\delta$ is homologous with $H$.  Moreover,
the previously discussed uniqueness of shortest paths implies that the weight of $H\oplus\delta = H \cup \sigma(v_i,u_i) \setminus \pi_i$ is smaller than the weight of~$H$.  But this contradicts our assumption that $H $ has minimum weight in its homology class.

Now suppose some pair of loops $\ell_i$ and $\ell_j$ are homotopic; by definition, the cycle $\ell_i\cdot\reverse{\ell_j}$ is contractible.  These two loops are contractions of paths $\pi_i$ and $\pi_j$ in $G$ with endpoints in $\sigma$.  Let $u_i$ and~$v_i$ denote the endpoints of $\pi_i$, and let $u_j$ and~$v_j$ denote the endpoints of $\pi_j$.  The cycle $\pi_i \cdot \sigma(v_i,v_j) \cdot \overline{\pi_j} \cdot \sigma(u_j, u_i)$ in $G$ is also contractible.  Let $\delta$ denote the set of edges of $G$ that appear in this cycle exactly once.  If the sub-paths $\sigma(v_i,v_j)$ and $\sigma(u_j, u_i)$ are edge-disjoint, then $\delta$ is a contractible cycle; otherwise, $\delta$ is the union of two non-crossing homotopic cycles.  In either case, $\delta$ is a boundary subgraph, so the symmetric difference $H\oplus\delta$ is homologous with $H$.  Moreover, $H\oplus\delta$ has smaller weight than~$H$, and we obtain another contradiction.

We conclude that the loops $\ell_1, \ell_2, \dots, \ell_x$ lie in distinct nontrivial homotopy classes.  Thus, these loops define an embedding of a single-vertex graph with $x$ edges onto $\Sigma$, where every face of the embedding is bounded by at least three edges.  Euler's formula now implies that $x\le 6g+2b-3$~\cite[Lemma~2.1]{ccelw-scsih-08}.
\end{proof}

We emphasize that different cycle decompositions of the same even subgraph $H'$ may lead to
different numbers of crossings.  Our crossing bound applies to \emph{every} cycle decomposition of
$H'$.

\subsection{Triangulations and crossing sequences}
\label{SS:homotopy-triangulation}

We can now describe our algorithm.  First, recall that we have a combinatorial surface $\Sigma$ along with an associated greedy system of arcs on that surface; if  $\Sigma$ had no boundary, we can delete a face without loss of generality and proceed, as described in Section~\ref{S:tree-cotree}.
We will cut the combinatorial surface $\Sigma$ along our greedy system of arcs into a $2\beta$-gon, or \emph{abstract polygonal schema}.  This construction cuts along each path in the greedy system; we then have two copies of each path,  and replace each copy  with a single edge.  Thus, each path in our greedy system of arcs will correspond to two edges in the polygon.

We next dualize the abstract polygonal schema by replacing each edge with a vertex, and connecting vertices which correspond to adjacent edges in the primal schema.  Any collection of non-crossing, non-self-crossing cycles corresponds to a \emph{weighted triangulation} \cite{ccelw-scsih-08}, where we draw an edge between two vertices of the dual abstract polygonal schema if and only if some cycle consecutively crosses the corresponding pair of paths in the greedy system of loops.  Each edge is weighted by the number of times such a crossing occurs in our collection.  Conversely, a weighted triangulation corresponds to a collection of non-crossing, non-self-crossing cycles as long as corresponding vertices are incident to edges of equal total weight.  See Figure~\ref{fig:weightedtriangulation} for an illustration of this correspondence.  Lemma~\ref{lem:crossing} implies that we only need to consider weights between 0 and $O(g+b)$.  Thus, there are $(g+b)^{O(g+b)}$ different weighted triangulations for each valid crossing vector.

\begin{figure}[htb]
\centering\includegraphics[height=1.45in]{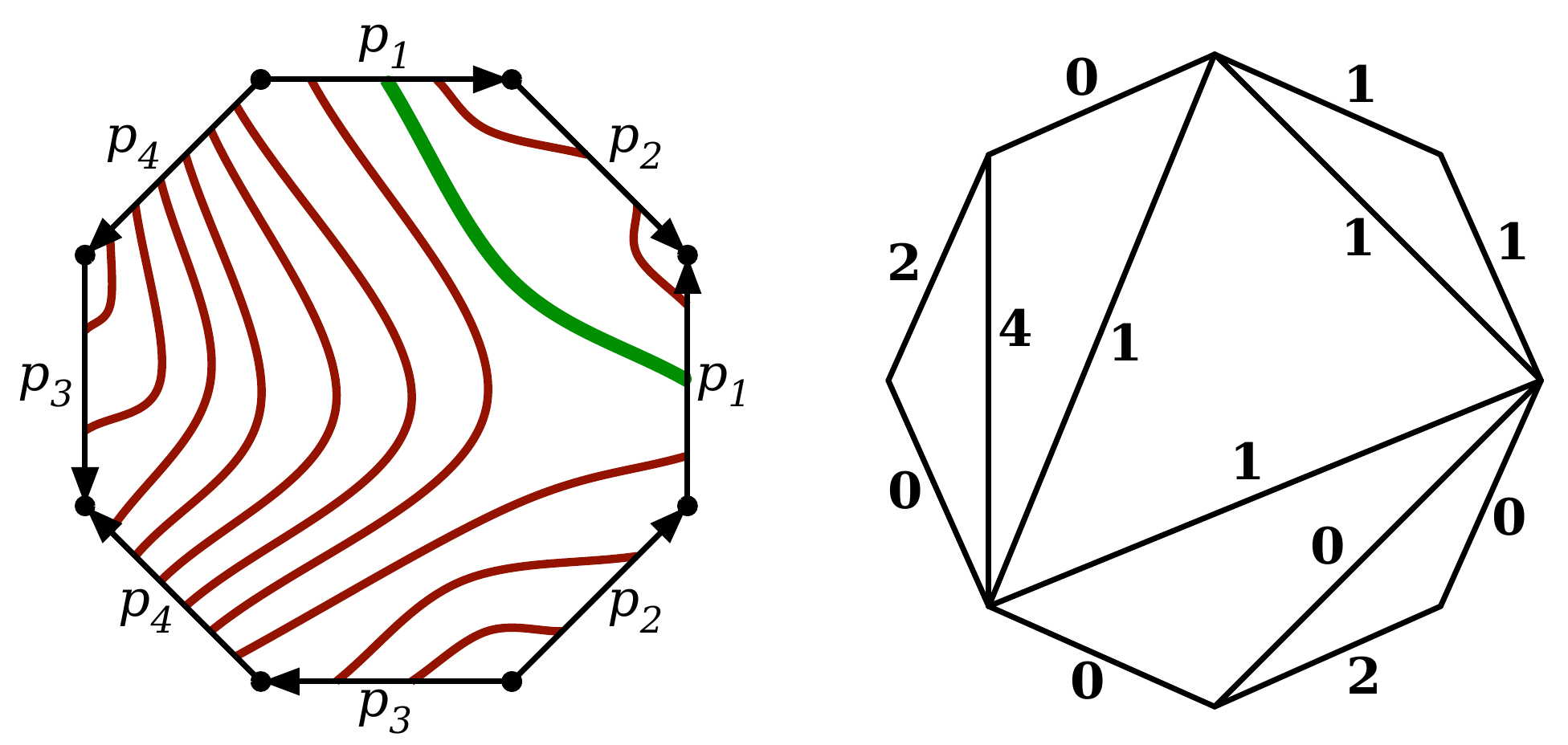}
\caption{Two disjoint simple cycles on a surface of genus 2, and the corresponding weighted triangulation.}
\label{fig:weightedtriangulation}
\end{figure}

For each valid weighted triangulation, we can compute a corresponding collection of abstract cycles in $O((g+b)^2)$ time by brute force.  In the same time, we can also compute the \emph{sequence} of crossings of each abstract cycle with the paths from our greedy system of arcs.  An algorithm of Kutz~\cite{k-csnco-06} computes the shortest cycle in $G$ with a given crossing sequence of length $x$ in $O(x n \log n)$ time; this proceeds by slicing $\Sigma$ along the greedy system of arcs, and then gluing together $x$ copies of the resulting planar surface into an annulus and calling Frederickson's planar minimum-cut algorithm \cite{f-faspp-87}.
Italiano \etal~\cite{insw-iamcmf-11} point out that their recent $O(n \log \log n)$-time improvement in computing minimum $(s,t)$-cuts in planar graphs can be used instead of Frederickson's algorithm. 
Thus, for each weighted triangulation, we obtain the shortest corresponding set of cycles in $O((g+b)^2 n \log \log n)$ time.

\begin{theorem}
\label{Th:Z2-minimal-crossing}
Let $G$ be an undirected graph with positively weighted edges embedded on a surface with genus $g$ and $b$ boundary components, and let $H$ be an even subgraph of $G$. We can compute the minimum-weight even subgraph homologous with $H$ in $(g+b)^{O(g+b)} n\log \log n$ time.
\end{theorem}

\begin{corollary}
\label{C:min-cut-crossing}
Let $G$ be an edge-weighted undirected graph embedded on a surface with genus $g$ and $b$ boundary components, and let $s$ and $t$ be vertices of $G$.  We can compute the minimum-weight $(s,t)$-cut in $G$ in $g^{O(g)} n\log \log n$ time.
\end{corollary}

\section{The $\Z_2$-Homology Cover}
\label{sec:homcover}

At a very high level, our algorithms in Section 4 find the minimum-weight subgraph in a given homology class by enumerating possible \emph{homotopy} classes of the cycles in a cycle decomposition, and then finding the shortest cycle in each possible homotopy class by searching a finite portion of the universal cover of the surface~$\Sigma$.  In this section, we describe a more direct algorithm, which finds the shortest cycle in each \emph{homology} class, by constructing and searching a space which we call the \emph{$\Z_2$-homology cover}.
Specifically, given a homology signature $h\in (\Z_2)^\beta$, our algorithm computes the shortest
cycle with signature $h$ in $\beta^{O(\beta)} n \log n$ time, using a generalization of Klein's
multiple-source shortest path algorithm~\cite{k-msspp-05} for planar graphs to higher-genus embedded
graphs \cite{cce-msspe-13,efl-hmcpf-18}.
In fact, because there are only $2^\beta$ homology classes, we can compute the shortest cycle in \emph{every} homology class in the same running time.
We then assemble the minimum-weight \emph{even subgraph} in any given homology class from these $\Z_2$-minimal cycles using dynamic programming.

In the preliminary version of this section~\cite{en-mcsnc-11}, we described an algorithm to compute
the shortest non-separating cycle in a directed surface graph in $g^{O(g)} n \log n$ time, improving
(for fixed $g$) an earlier algorithm of Cabello \etal~\cite{ccl-fsncd-10} that runs in
$O(g^{1/2}n^{3/2}\log n)$ time.  Using similar techniques but with different covering spaces,
Erickson~\cite{e-sncds-11} and Fox~\cite{f-sntcd-13} described even faster algorithms that find
shortest non-separating cycles in~$O(g^2 n \log n)$ time and shortest non-contractible cycles
in~$O(g^3 n \log n)$ time.  In light of these improvements, we omit discussion of our non-separating
cycle algorithm from this paper.

%
%
%
%
%
%
%
%

\subsection{Definition and construction}
\label{sec:homcover_cover}

We begin by computing homology signatures for the edges of $G$ in $O(\beta n)$ time, as described in Section~\ref{sec:characterizing_signatures}.
After computing homology signatures for each edge, the $\Z_2$-homology cover of a combinatorial surface can be defined using a standard~\emph{voltage construction}~\cite[Chapter 4]{gt-tgt-01}, as follows.

We first define the covering graph $\Gbar$.  For simplicity, we regard every edge $uv$ of $G$ as a pair of oppositely oriented \emph{darts} $\arc{u}{v}$ and $\arc{v}{u}$.  The vertices of $\Gbar$ are all ordered pairs $(v, h)$ where $v$ is a vertex of $G$ and $h$ is an element of $(\Z_2)^\beta$.  The darts of $\Gbar$ are the ordered pairs $(\arc{u}{v}, h) := (u, h)\arcto(v, h\oplus [uv])$ for all edges $\arc{u}{v}$ of $G$ and all homology classes $h \in (\Z_2)^\beta$, and the reversal of any dart $(\arc{u}{v}, h)$ is the dart $(\arc{v}{u}, h\oplus[uv])$.

Now let $\pi\colon \Gbar\to G$ denote the covering map $\pi(v, h) = v$; this map projects any cycle in $\Gbar$ to a cycle in $G$.  To define a cellular embedding of~$\Gbar$, we declare a cycle in $\Gbar$ to be a face if and only if its projection is a face of $G$.  The combinatorial surface defined by this embedding is the $\Z_2$-homology cover $\Sigmabar$.

Our construction can be interpreted more topologically as follows.  Let $\dualarc_1, \dots, \dualarc_\beta$ denote the system of dual arcs used to define the homology signatures $[e]$.  The surface $D := \Sigma\setminus(\dualarc_1\cup\cdots\cup \dualarc_\beta)$ is a topological disk.  Each arc $\dualarc_i$ appears on the boundary of $D$ as two segments $\dualarc^+_i$ and~$\dualarc^-_i$.  For each signature $h\in (\Z_2)^\beta$, we create a disjoint copy $(D,h)$ of $D$; for each index~$i$, let $(\dualarc^+_i, h)$ and $(\dualarc^-_i, h)$ denote the copies of $\dualarc^+_i$ and $\dualarc^-_i$ in the disk $(D, h)$.  For each index~$i$, let~$b_i$ denote the $\beta$-bit vector whose $i$th bit is equal $1$ and whose other $\beta-1$ bits are all equal to~$0$.  The $\Z_2$-homology cover $\Sigmabar$ is constructed by gluing the~$2^\beta$ copies of $D$ together by identifying boundary paths $(\dualarc^+_i,h)$ and $(\dualarc^-_i, h\oplus b_i)$, for every index $i$ and homology class $h$.  See Figure \ref{fig:cover-ex} for an example.

\begin{figure}
\centering
\includegraphics[height=1.5in]{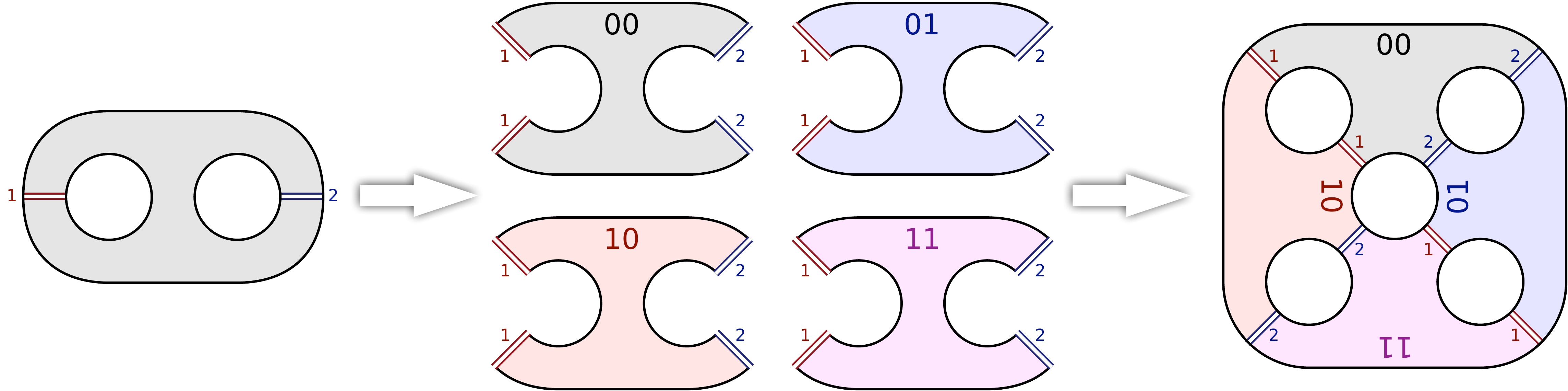}
\caption{Constructing the $\Z_2$-homology cover of a pair of pants (a genus zero surface with three boundaries).}
\label{fig:cover-ex}
\end{figure}

\begin{lemma}
\label{lem:cover-cxy}
The combinatorial surface $\Sigmabar$ has $\nbar = 2^\beta n$ vertices, genus $\gbar = O(2^\beta \beta)$, and $\bbar = O(2^\beta b)$ boundaries, and it can be constructed in $O(2^\beta n)$ time.
\end{lemma}

\begin{proof}
Let $m$ and $f$ denote the number of edges and faces of $\Sigma$, respectively.  Recall that the Euler characteristic of $\Sigma$ is $\chi = n - m + f = 2 - 2g - b = 1-\beta$.  The combinatorial surface~$\Sigmabar$ has exactly $\nbar = 2^\beta n$ vertices, $2^\beta m$ edges, and $2^\beta f$ faces, so its Euler characteristic is $\chibar = 2^\beta (1-\beta)$.

If $b>1$, then each boundary cycle $\delta_i$ has a non-zero homology signature; at least one arc $\dualarc_j$ has exactly one endpoint on~$\delta_i$.  Thus, $\Sigmabar$ has exactly $\bbar = 2^{\beta-1} b$ boundary cycles, each of which is a double-cover (in fact, the $\Z_2$-homology cover) of some boundary cycle~$\delta_i$.  It follows that~$\Sigmabar$ has genus $\gbar = 1-(\chibar+\bbar)/2 = {2^{\beta-2} ({4g+b-4}) + 1}$.  (Somewhat surprisingly, $\Sigmabar$ may have positive genus even when $\Sigma$ does not!)  On the other hand, when $b=1$, the boundary cycle~$\delta_1$ is null-homologous, so $\Sigmabar$ has $\bbar = 2^\beta b$ boundary cycles, and thus~$\Sigmabar$ has genus $\gbar = {1-(\chibar+\bbar)/2} =  {2^\beta (g-1) + 1}$.

After computing the homology signatures for $\Sigma$ in $O(\beta n)$ time, following Lemma \ref{lem:sign}, it is straightforward to construct $\Sigmabar$ in $O(\nbar) = O(2^\beta n)$ time.
\end{proof}

Each edge in $\Gbar$  inherits the weight of its projection in  $G$.
Now consider an arbitrary path $\path$ in $G$, with (possibly equal) endpoints $u$~and $v$.  A straightforward induction argument implies that for any homology class $h \in (\Z_2)^\beta$, the path~$\path$ is the projection of a unique path from $(u,h)$ to $(v,h\oplus[\path])$, which we denote \EMPH{$(\path,h)$}.  Moreover, this lifted path has the same length as its projection.  The following lemmas are now immediate.

\begin{lemma}
\label{lem:lift-shortest}
Every lift of a shortest path in $G$ is a shortest path in $\Gbar$.
\end{lemma}

\begin{lemma}
\label{lem:lift-minimal}
A loop $\ell$ in $G$ with basepoint $v$ is $\Z_2$-minimal if and only if, for every homology class $h\in (\Z_2)^\beta$, the lifted path $(\ell,h)$ is a shortest path in $\Gbar$ from $(v,h)$ to $(v,h\oplus[\ell])$.
\end{lemma}

\subsection{Computing $\Z_2$-minimal cycles}
\label{sec:homcover_cycles}

The results in the previous section immediately suggest an algorithm to compute the shortest cycle in a given $\Z_2$-homology class $h$ in time $2^{O(\beta)}n^2$: construct the $\Z_2$-homology cover, and then compute the shortest path from $(v,0)$ to $(v,h)$, for every vertex $v$ in the original graph.  In this section, we describe a more complex algorithm that runs in time $2^{O(\beta)}n\log n$.
Recall that any path $\sigma$ from $u$ to $v$ in $G$ is the projection of a unique path $(\sigma,0)$ from $(u,0)$ to $(v,[\sigma])$ in $\Gbar$.

\begin{lemma}
\label{lem:nocross}
Let $\cycle$ be a $\Z_2$-minimal cycle in $G$, and let~$\sigma$ be any shortest path in $G$ that intersects $\cycle$.  There is a $\Z_2$-minimal cycle $\cycle'$ homologous to $\cycle$, which is the projection of a shortest path $(\cycle', h)$ in $\Gbar$ that starts with a subpath of $(\sigma, 0)$ but does not otherwise intersect $(\sigma, 0)$.
\end{lemma}

\begin{proof}
Let $v$ be the vertex of $\sigma\cap\cycle$ closest to the starting vertex of $\sigma$, and let $(v,h)$ be the corresponding vertex of the lifted path $(\sigma,0)$.  Think of~$\cycle$ as a loop based at $v$.  Lemma \ref{lem:lift-minimal} implies that the lifted path $(\cycle, h)$ is a shortest path from $(v,h)$ to $(v,h\oplus [\cycle])$.

Now let $(w,h')$ be the last vertex along $(\cycle,h)$ that is also a vertex of $(\sigma,0)$.  Let $(\cycle', h)$ be the path obtained from $(\cycle, h)$ by replacing the subpath from from $(v,h)$ to $(w,h')$ with the corresponding subpath of $(\sigma,0)$.  By construction, $(\cycle', h)$ starts with a subpath of $(\sigma,0)$ but does not otherwise intersect $(\sigma,0)$.   Because both $(\cycle, h)$ and $(\sigma,0)$ are shortest paths in $\Sigmabar$, the new path $(\cycle', h)$ has the same length as $(\cycle, h)$.  Thus, the projected cycle $\cycle'$ has the same length and homology class as $\cycle$, which implies that~$\cycle'$ is $\Z_2$-minimal.
\end{proof}

We emphasize that the modified cycle $\cycle'$ may intersect~$\sigma$ arbitrarily many times; however, all such intersections lift to intersections between $(\cycle', h)$ and lifts of $\sigma$ other than $(\sigma, 0)$.

Our algorithm uses a generalization of Klein's  multiple-source shortest path algorithm~\cite{k-msspp-05} to higher-genus embedded graphs:

\begin{lemma}[Erickson \etal~\cite{efl-hmcpf-18}]
\label{lem:multishort}
Let $G$ be a graph with non-negative edge weights, cellularly embedded on a surface $\Sigma$ of
genus $g$ with $b>0$ boundaries, and let $f$ be an arbitrary face of~$G$.
We can preprocess $G$ in $O(g^2n\log n)$ time and $O(g n \log n)$ space so that the length of the
shortest path from any vertex incident to $f$ to any other vertex can be retrieved in $O(\log n)$
time.
\end{lemma}

The above result is based on a procedure of Chambers \etal~\cite{cce-msspe-13} which requires only
$O(g n \log n)$ preprocessing time but relies on the shortest path between any pair of vertices
being unique.
We can enforce uniqueness of shortest paths and achieve the $O(g n \log n)$ preprocessing time
\emph{with high probability} by adding random infinitesimal weights to each edge \cite{mvv-memi-87}.
Because our minimum-weight $(s,t)$-cut algorithm is otherwise deterministic, we will use the
deterministic $O(g^2 n \log n)$ preprocessing time procedure referenced in
Lemma~\ref{lem:multishort} throughout the rest of this section.

\begin{theorem}
\label{thm:min-cycle}
Let $G$ be a  graph with non-negative edge weights, cellularly embedded on a surface~$\Sigma$ with first Betti number $\beta$, and let~$\cycle$ be a cycle in~$G$ with $k$ edges.  A shortest cycle in~$\Sigma$ that is $\Z_2$-homologous with~$\cycle$ can be computed in $O(\beta k + 8^\beta \beta^3\, n\log n)$ time.
\end{theorem}

\begin{proof}
We begin by computing homology signatures for the edges of $G$ in $O(\beta n)$ time, as described in Section~\ref{sec:characterizing_signatures}.  In $O(\beta k)$ time, we then compute the homology signature~$[\cycle]$.  If $[\cycle] = 0$, we return the empty walk and halt.

Next, we construct the $\Z_2$-homology cover $\Gbar$ in $O(2^\beta n\log n)$ time, as described in Section~\ref{sec:homcover_cover}, as well as the set $S$ of directed shortest paths described in Section~\ref{sec:characterizing_crossings}.
Any cycle homologous with~$\cycle$ must intersect at least one member of~$S$ as~$\cycle$ is not null-homologous.
We look for the shortest path in $\Gbar$ of the canonical form described in Lemma~\ref{lem:nocross}, by considering each shortest path $\sigma\in S$ in turn as follows.

Let us write $(\sigma,0) = (v_0,0) \arcto (v_1,h_1) \arcto\allowbreak \cdots\arcto (v_t, h_t)$.  We construct the combinatorial surface $\Sigmabar\snip(\sigma,0)$ by splitting the path $(\sigma,0)$ into two parallel paths from $(v_0,0)$ to $(v_t,h_t)$, which we denote  $(\sigma,0)^+$ and $(\sigma,0)^-$.  For each index $1\le i\le t-1$, let $(v_i,h_i)^+$ and $(v_i,h_i)^-$ denote the copies of vertex $(v_i,h_i)$  on the paths $(\sigma,0)^+$ and $(\sigma,0)^-$, respectively.  The paths $(\sigma,0)^+$ and $(\sigma,0)^-$ bound a new common face $f_{(\sigma,0)}$ in $\Sigmabar\snip(\sigma,0)$.

Lemma \ref{lem:nocross} implies that if any $\Z_2$-minimal cycle homologous to $\cycle$ intersects $\sigma$, then some $\Z_2$-minimal cycle homologous to $\cycle$ is the projection of a shortest path in $\Sigmabar\snip(\sigma,0)$ from some vertex $(v_i,h_i)^\pm$ to the corresponding vertex $(v_i, h_i\oplus[\cycle])$.  To compute these shortest paths, we implicitly compute the shortest path in $\Sigmabar\snip(\sigma,0)$ from every vertex on the boundary of $f_{(\sigma,0)}$ to \emph{every} vertex of $\Sigmabar\snip(\sigma,0)$, using Lemma~\ref{lem:multishort}.
The resulting algorithm runs in $O(\gbar^2\,\nbar \log \nbar) = O(8^\beta \beta^3\, n\log n)$ time, by Lemma \ref{lem:cover-cxy}.
\end{proof}

By running this algorithm $2^\beta$ times, we can compute the shortest cycle in $\Sigma$ in
\emph{every} $\Z_2$-homology class, in $O(16^\beta \beta^3\, n\log n)$ time.

We note that this result holds in directed graphs as well, as the construction of the $\Z_2$-homology cover ignores weights and the multiple source shortest path algorithm works on directed graphs.

\subsection{Minimum cuts from the homology cover}
\label{sec:homcover_mincut}

We now apply our algorithm for computing $\Z_2$-minimal cycles to the problem of computing $\Z_2$-minimal even subgraphs in undirected surface embedded graphs.
Theorem~\ref{thm:min-cycle} immediately implies that we can compute a minimum-weight \emph{cycle} in
every $\Z_2$-homology class in $O(16^\beta \beta^3\, n\log n)$ time.  However, the minimum weight \emph{even subgraph} in a given homology class may not be (the carrier of) a $\Z_2$-minimal cycle.  In particular, if a $\Z_2$-minimal cycle $\gamma$ traverses any edge more than once, then every minimum-weight even subgraph with signature $[\gamma]$ \emph{must} be disconnected.

However, any \emph{connected} $\Z_2$-minimal even subgraph is the carrier of a $\Z_2$-minimal cycle, and the components of any $\Z_2$-minimal even subgraph are themselves $\Z_2$-minimal even subgraphs.  Thus, we can assemble a $\Z_2$-minimal even subgraph in any homology class from a subset of the $\Z_2$-minimal cycles we have already computed.  The following lemma puts an upper bound on the number of cycles we need.

\begin{lemma}
\label{lem:even-comps}
Every $\Z_2$-minimal even subgraph of $G$ has at most $g+b-1$ components.
\end{lemma}

\begin{proof}
Let $\gamma_1, \dots, \gamma_{g+b}$ be disjoint simple cycles on an \emph{abstract} surface $\Sigma$ of genus $g$ with $b$ boundaries, and consider the surface $\Sigma' = \Sigma \setminus (\gamma_1 \cup \cdots \cup \gamma_{g+b})$.  The definition of genus implies that $\Sigma'$ cannot be connected; indeed, $\Sigma'$ must have at least $b+1$ components.  So the pigeonhole principle implies that some component $\Sigma''$ of~$\Sigma$ does not contain any of the boundary cycles of $\Sigma$.  The boundary of $\Sigma''$ is therefore null-homologous.

Now let $H$ be an even subgraph of $G$ with more than $g+b-1$ components.  Each component has a cycle decomposition, so $H$ must have a cycle decomposition with more than $g+b-1$ elements.  Thus, the argument in the first paragraph implies that some subgraph of $H$ must be null-homologous.  We conclude that $H$ is not $\Z_2$-minimal.
\end{proof}

\begin{theorem}
\label{thm:min-even}
Let $G$ be an undirected graph with non-negative edge weights, cellularly embedded on a
surface~$\Sigma$ with first Betti number $\beta$.  A minimum-weight even subgraph of $G$ in each
$\Z_2$-homology class can be computed in $O(16^\beta \beta^3\, n\log n)$ time.
\end{theorem}

\begin{proof}
Our algorithm computes a minimum-weight cycle~$\gamma_h$ in every $\Z_2$-homology class~$h$ in $O(16^\beta \beta^3\, n\log n)$ time, via Theorem \ref{thm:min-cycle}, and then assembles these $\Z_2$-minimal cycles into $\Z_2$-minimal even subgraphs using dynamic programming.

For each homology class $h\in (\Z_2)^\beta$ and each integer $1\le k\le g+b-1$, let \EMPH{$C(h,k)$} denote the minimum total weight of any set of at most $k$ cycles in $G$ whose homology classes sum to $h$.  Lemma~\ref{lem:even-comps} implies that the minimum weight of any even subgraph in homology class~$h$ is exactly $C(h, g+b-1)$.  This function obeys the following straightforward recurrence:
\[
	C(h,k) = \min\Setbar{C(h_1,k-1) + C(h_2, 1)\strut}{h_1\oplus h_2 = h}.
\]
This recurrence has two base cases: $C(0, k) = 0$ for any integer $k$, and for any homology class $h$, the value $C(h,1)$ is just the length of $\gamma_h$.  A standard dynamic programming algorithm computes $C(h, g+b-1)$ for all $2^\beta$ homology classes $h$ in $O(4^\beta \beta)$ time.  We can then assemble the actual minimum-weight even subgraphs in each homology class in $O(\beta n)$ time.  The total time for this phase of the algorithm is $O(4^\beta \beta + 2^\beta \beta n)$, which is dominated by the time to compute all the $\Z_2$-minimal cycles.
\end{proof}

\begin{corollary}
\label{cor:mincut}
Let $G$ be an edge-weighted undirected graph embedded on a surface with genus $g$ and $b$ boundary
components, and let $s$ and $t$ be vertices of $G$.  We can compute the minimum-weight $(s,t)$-cut
in $G$ in $O(256^g g^3 n \log n)$ time.
\end{corollary}

As discussed above, we can instead randomly perturb edge weights and use the faster multiple-source
shortest paths algorithm of Chambers \etal~\cite{cce-msspe-13} to find a minimum-weight $(s,t)$-cut
with high probability in $O(64^g g^3 n \log n)$ time.

\section{{NP}-Hardness}
\label{S:NPhard}

In this section, we show that finding the minimum-cost even subgraph in a given homology class is {NP}-hard, even when the underlying surface has no boundary.  Our proof closely follows a reduction of McCormick \etal~\cite{mrr-edofm-03} from \textsc{Min2Sat} to a special case of \textsc{MaxCut}.

\begin{theorem}
Computing the minimum-weight even subgraph in a given homology class on a surface without boundary is equivalent to computing a minimum-weight cut in an embedded edge-weighted graph $G$ whose negative-weight edges are dual to an even subgraph in $G^*$.
\end{theorem}

\begin{proof}
Fix a graph $G$ embedded on a surface $\Sigma$ without boundary, together with a weight function $c\colon E\to \Real$.  For any even subgraph $H$ of $G$, let $c(H) = \sum_{e\in H} c(e)$, and let $\textsc{MinHom}(H,c)$ denote the even subgraph of minimum weight in the homology class of $H$.

Consider the \emph{residual weight} function $c_H\colon E\to \Real$ defined by setting $c_H(e) = c(e)$ for each edge $e\not\in H$, and $c_H(e) = -c(e)$ for each edge $e\in H$.  For any subgraph $H'$ of $G$, we have $c(H') = c_H(H\oplus H') + c(H)$, which immediately implies that
\(
    \textsc{MinHom}(H,c) ~=~
    H \oplus \textsc{MinHom}(\varnothing, c_H).
\)

Every boundary subgraph of $G$ is dual to a cut in the dual graph $G^*$.  Thus, we have reduced our problem to computing the minimum cut in $G^*$ with respect to the weight function $c_H$.  Since the empty set is a valid cut with zero cost, the cost of the minimum cut is never positive.  In particular, $H$ is the minimum-cost even subgraph in its homology class if and only if the cut in $G^*$ with minimum residual cost is empty.

In fact, our reduction is reversible.  Suppose we want to find the minimum cut in an embedded graph $G = (V, E)$ with respect to the cost function $c\colon E\to \Real$, where every face of $G$ is incident to an even number of edges with negative cost.  Let $H = \set{{e\in E}\mid {c(e)<0}}$ be the subgraph of negative-cost edges, and let $X$ denote the (possibly empty) set of edges in the minimum cut of $G$.  Consider the \emph{absolute cost} function $\abs{c}\colon E^*\to \Real$ defined as $\abs{c}(e^*) = \abs{c(e)}$.  Then $(H\oplus X)^*$ is the even subgraph of $G^*$ of minimum absolute cost that is homologous to $H^*$.
\end{proof}

We now prove that this special case of the minimum cut problem is {NP}-hard, by  reduction from \textsc{MinCut} in graphs with negative edges.  This problem includes \textsc{MaxCut} as a special case (when every edge has negative cost), but many other special cases are also {NP}-hard~\cite{mrr-edofm-03}.  The output of our reduction is a simple triangulation; the reduction can be simplified if graphs with loops and parallel edges are allowed.

Suppose we are given an \emph{arbitrary} graph $G = (V,E)$ with $n$ vertices and an \emph{arbitrary} cost function $c\colon E\to \Real$.  We begin by computing a cellular embedding of $G$ on some orientable surface, by imposing an arbitrary cyclic order on the edges incident to each vertex.
(We can compute the maximum-genus orientable cellular embedding  in polynomial time~\cite{fgm-fmggi-88}.)  Alternately, we can add zero-length edges to make the graph complete and then use classical results of Ringel, Youngs, and others \cite{ry-shmcp-68,r-mct-74} to compute a minimum-genus orientable embedding of $K_n$ in polynomial time.  Once we have an embedding, we add vertices and zero-cost edges to obtain a triangulation.

Let $C$ be the sum of the absolute values of the edge costs: $C:= \sum_e \abs{c(e)}$.
A \EMPH{cocycle} of embedded graph $G$ is a subset of edges forming a cycle in the dual $G^*$.
We locally modify both the surface and the embedding to transform each negative-weight edge into a cocycle, as follows.  We transform the edges one at a time; after each iteration, the embedding is a simple triangulation.  (Our reduction can be simplified if a simple graph is not required.)  For each edge~$uv$ with $c(uv)<0$, remove $uv$ to create a quadrilateral face.  Triangulate this face as shown in Figure \ref{fig:addhandle}; we call the new faces $uu_1u_2$ and $vv_1v_2$ \emph{endpoint triangles}.  Assign cost $C$ to the edges of the endpoint triangles and cost zero to the other new edges. Glue a new handle to the endpoint triangles, and triangulate the handle with a cycle of six edges, each with cost $c(uv)/6$.  These six edges form a cocycle of cost $c(uv)$, which we call an \emph{edge cocycle}, in the new embedding.  Each iteration adds $5$ vertices and $21$ edges to the graph and increases the genus of the underlying surface by $1$.

\begin{figure}[hbt]
\centering\includegraphics[scale=0.4]{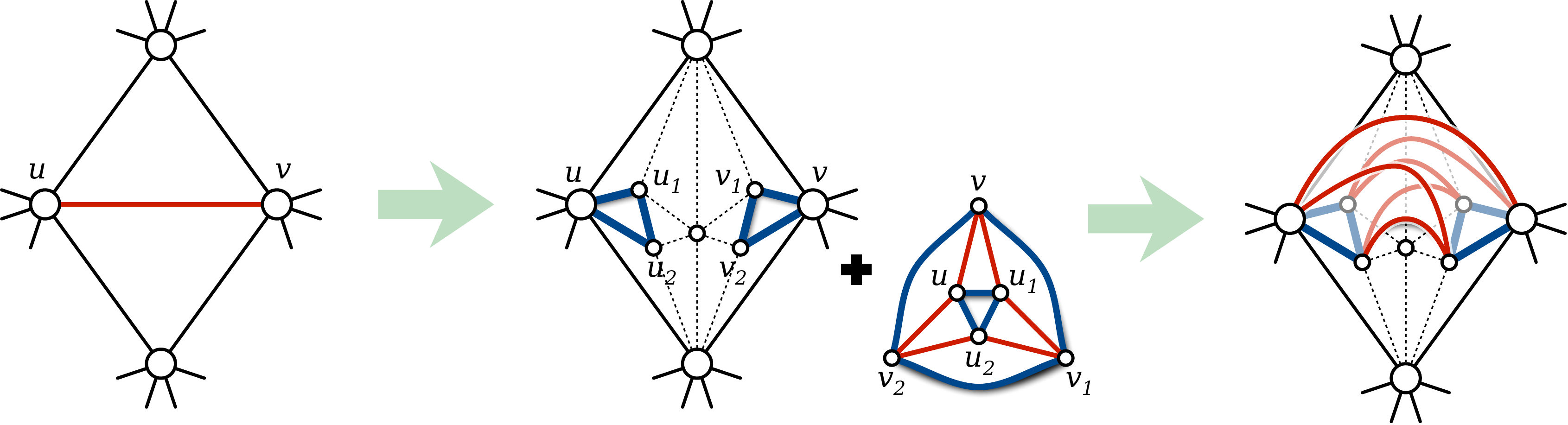}
\caption{Adding a handle to transform a negative edge into a negative cocycle.  Thick (blue) edges have cost $C$; dashed edges have cost zero.}
\label{fig:addhandle}
\end{figure}

Let $G'$ denote the transformed graph and $c'\colon E(G')\to \Real$ its associated cost function.  The minimum cut in $G'$ cannot contain any edge of an endpoint triangle.  Thus, for each edge cocycle, either all six edges cross the cut, or none of them cross the cut.  It follows that the minimum cut in $G'$ corresponds to a cut with equal cost in the original graph $G$.  Conversely, any cut in $G$ can be transformed into a cut in~$G'$ of equal cost.  Thus, computing the minimum cut in $G'$ is equivalent to computing the minimum cut in $G$.

\begin{theorem}
Given an even subgraph $H$ of an edge-weighted graph $G$ embedded on a surface without boundary, computing the minimum-weight even subgraph homologous to $H$ is strongly {NP}-hard.
\end{theorem}

Our reduction can be modified further to impose other desirable properties on the output instances, for example, that the graph is unweighted, every vertex has degree $3$, or the input subgraph $H$ is a simple cycle.

Finally, we emphasize that the {NP}-hardness of this problem relies crucially on the fact that we are using homology with coefficients taken from the finite field $\Z_2$.  The corresponding problem for homology with real or integer coefficients is a minimum-cost circulation problem, and thus can be solved in polynomial time.
Chambers, Erickson, and Nayyeri~\cite{cen-hfcc-12} show that this circulation problem can be solved in near-linear time for graphs of constant genus and polynomially bounded integer edge capacities using very different techniques.

\def\minSS{X}

\section{Global Minimum Cut}
\label{sec:global}

Finally, we describe our algorithm to compute \emph{global} minimum cuts in surface-embedded graphs, where no source and target vertices are specified in advance.  
Unlike previous sections, we begin our exposition assuming that the underlying surface of the input
graph does \emph{not} have boundary, because filling in any boundaries with disks does not change
the minimum cut.
We also assume without loss of generality that no edge of the input graph has the same face on both
sides; we can enforce this assumption by adding zero-weight edges if necessary.

As in previous sections, it is convenient to work in the dual graph.  We cannot apply Lemma \ref{lem:cut-duality} directly, but the following lemma similarly characterizes global minimum cuts in surface graphs in terms of homology in the dual graph.
Suppose we have a graph embedded in a surface with a single boundary component.
A \EMPH{separating subgraph} is any non-empty boundary subgraph, or equivalently, the boundary of the union of a \emph{non-empty} set of faces.

\begin{lemma}
\label{lem:mincut-z2}
Let~$G$ be an undirected edge-weighted graph embedded on a surface~$\Sigma$ without boundary, and
let $s$ be an arbtirary vertex of $G$.
If~$\Cut$ is a global minimum cut in~$G$, then~$\Cut^*$ is a minimum-weight separating subgraph
of~$G^*$ in $\Sigma \setminus s^*$.
\end{lemma}

\begin{proof}
  Let~$\Cut$ be an arbitrary cut in~$G$.  The cut partitions the vertices of $G$
  into two disjoint subsets~$S$ and~$T$ with $s \in S$.
  Therefore, the dual subgraph~$\Cut^*$
  partitions the faces of~$G^*$ into two disjoint subsets~$S^*$ and~$T^*$ with $s^* \in S^*$.
  Further,~$\Cut^*$ is the boundary of the union of faces in~$T^*$, implying
  that~$\Cut^*$ is a boundary subgraph of~$\Sigma$ and therefore separating.

  Conversely, let~$\Cut^*$ be any separating subgraph of~$G^*$.
  Subgraph~$\Cut^*$ is the boundary of a nonempty subset of the faces~$T^*$
  of~$G^*$.
  Let $t^*$ be a face in~$T^*$.
  Any path from~$s$ to~$t$ in the primal
  graph~$G$ must traverse at least one edge of~$\Cut$.  We conclude that~$\Cut$ is
  a cut (in particular, an $(s,t)$-cut).
\end{proof}

In light of this lemma, the remainder of this section describes an algorithm to find a minimum-weight separating subgraph in a given surface-embedded graph $G$ with non-negative edge weights.
Graph $G$ is embedded in a surface $\Sigma$ with \emph{exactly one boundary component $s^*$}.
%


Let $\minSS$ be a minimum-weight separating subgraph.
Surface $\Sigma \setminus \minSS$ has exactly one component not incident to $s^*$; otherwise, the
boundary of any one of these components is a smaller separating subgraph.
Abusing terminology slightly,
call the separating subgraph $\minSS$ \EMPH{contractible} if this component of $\Sigma \setminus
\minSS$ is a disk, and \EMPH{non-contractible} otherwise.  If $\minSS$ is contractible, then
$\minSS$ is actually a shortest (weakly) simple contractible cycle of $G$ in the surface $\Sigma$; otherwise, $\minSS$ can be decomposed into one or more simple cycles, each of which is non-contractible.  See Figure~\ref{fig:global_cases}.

%
%
%

\begin{figure}[ht]
\centering
\includegraphics[height=1in]{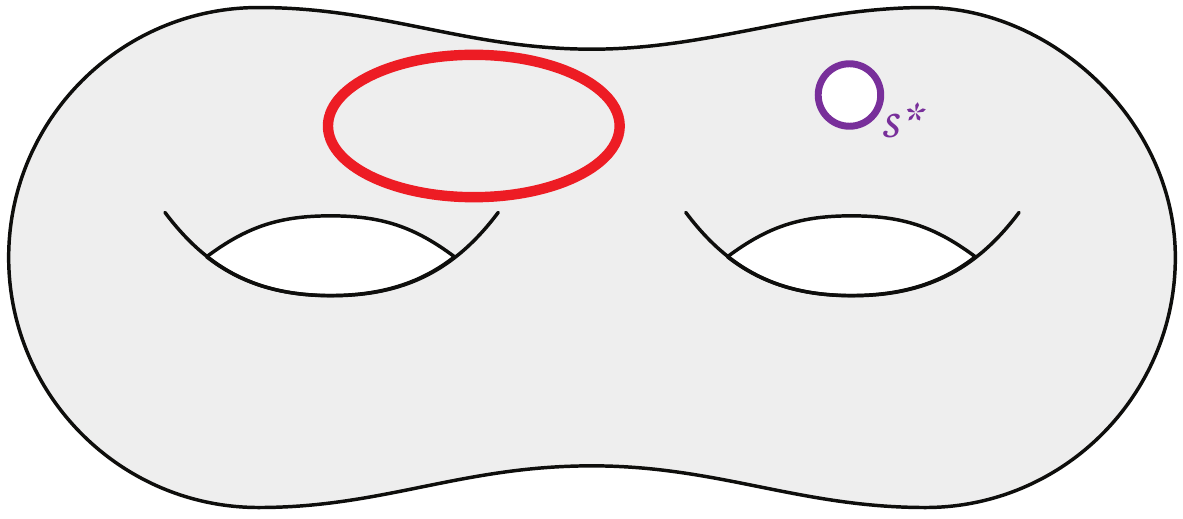}\qquad
\includegraphics[height=1in]{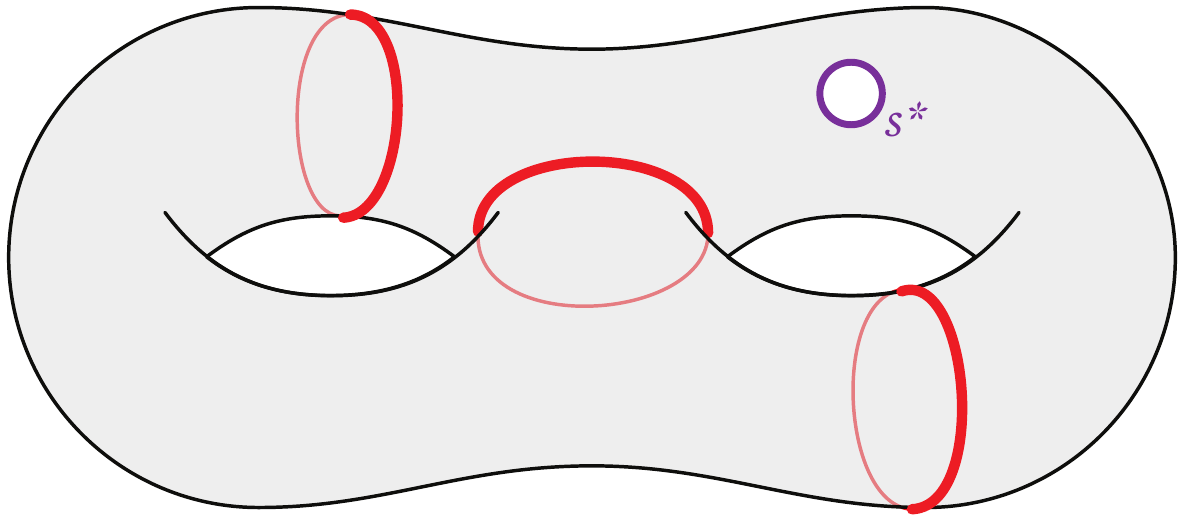}
\caption{Two types of minimum-weight separating subgraphs: a simple contractible cycle and otherwise.}
\label{fig:global_cases}
\end{figure}

Thus, in principle, we can find a minimum-weight separating subgraph by first computing a shortest contractible cycle, then computing a minimum-weight separating collection of non-contractible cycles, and finally returning the lighter of these two subgraphs.  Unfortunately, we do not know how to solve either of these subproblems in our stated time bounds, so our algorithm takes a more subtle approach. 


In Section \ref{sec:global_contractible}, we describe an algorithm that computes a minimum-weight separating subgraph \emph{if} any minimum-weight separating subgraph is contractible.  Similarly, in Section \ref{sec:global_non-contractible}, we describe an algorithm that computes a minimum-weight separating subgraph \emph{if} any minimum-weight separating subgraph is non-contractible.  In both cases, if no minimum-weight separating subgraph satisfies the corresponding condition, the algorithm still returns a boundary subgraph, but this subgraph could be empty or have large weight.  By running both subroutines and returning the best result, we are guaranteed to find a minimum-weight separating subgraph in $G$, no matter which category it falls into.


\subsection{Contractible}
\label{sec:global_contractible}

First we consider the case where some minimum-weight separating subgraph $\minSS$ is contractible.  
We begin by borrowing a result of Cabello~\cite[Lemma 4.1]{c-fscss-10}.  Recall that an arc or cycle is \emph{tight} if it has minimum-weight among all arcs or cycles in its homotopy class.

\begin{lemma}[Cabello~\cite{c-fscss-10}]
\label{lem:disjoint-tight-arc}
Let $\alpha$ be a tight arc or tight cycle on $G$.  There exists a shortest simple contractible cycle that does not cross $\alpha$.
\end{lemma}



Cabello~\cite{c-fscss-10} uses this observation to compute a shortest simple contractible cycle in a surface embedded graph; unfortunately, his algorithm runs in $O(n^2\log n)$ time.
Cabello \etal~\cite{cdem-fotc-10} use the same observations to find a shortest contractible \emph{closed walk} that encloses a non-empty set of faces in $O(n\log n)$ time (with no dependence on $g$).
%
%
However, the resulting closed walk $W$ might traverse some edges of~$G$ more than once.  The set of edges that $W$ traverses an odd number of times does constitute a boundary subgraph; unfortunately, this subgraph could be disconnected or even empty.

Our algorithm closely follows the algorithm of Cabello \etal~\cite{cdem-fotc-10}, but with modifications to ensure that the resulting walk traverses at least one edge an odd number of times.

We use the slicing operation ($\snip$) along tight cycles and arcs in~$G$.  The following lemma
implies it is safe for our algorithm to find minimum-weight separating subgraphs in sliced copies
of~$\Sigma$.

\begin{lemma}
\label{lem:global_null-homologous-projections}
Let~$\alpha$ be an arbitrary simple cycle or arc in~$G$.
Let
${\Sigmasnip = \Sigma \snip \alpha}$ and let~$\Gsnip = G \snip \alpha$. Any null-homologous closed
walk~$\gammasnip$ in~$\Gsnip$ projects  to a null-homologous closed walk in~$G$.
\end{lemma}

\begin{proof}
Let~$\gammasnip$ be a null-homologous closed walk in~$\Gsnip$ and let~$\gamma$ be its projection in~$G$.
Let $\Hsnip$ be the even subgraph of $\Gsnip$ containing exactly the edges that appear an odd number of
times in $\gammasnip$.
Let $H$ be the even subgraph of $G$ containing exactly the edges that appear an odd number of times
in $\gamma$.
Subgraph~$\Hsnip$ bounds a subset of faces~$\Fsnip$ in~$\Gsnip$. 
Let~$F$ be natural mapping of~$\Fsnip$ into~$G$.
We will argue that~$H$ is the boundary of~$F$, proving the lemma.

Consider any edge $e$ of $G$.
Suppose $e$ is not in $\alpha$.
In this case, $\Gsnip$ contains one copy~$e'$ of $e$, and all faces incident to~$e'$ map to faces
incident to~$e$.
Edge $e'$ being incident to exactly one face of $\Fsnip$ is equivalent to $e' \in H'$, which in turn
is equivalent to $\gammasnip$ using $e'$ an odd number of times, $\gamma$ using $e$ an odd number of
times, and finally $H$ containing $e$.

Now suppose $e$ is in $\alpha$.
Graph $\Gsnip$ contains two copies of $e$ denoted $e_1$ and $e_2$ each incident to one face $f_1$
and $f_2$, respectively (note that $f_1$ and $f_2$ may be equal).
If neither or both of $f_1$ and $f_2$ are in $\Fsnip$, then $\Hsnip$ includes neither or both of
$e_1$ and $e_2$.
In turn, $\gammasnip$ goes through the two copies of $e$ an even number of times total, meaning
$\gamma$ uses $e$ an even number of times and $e \notin H$.
If one, but not both, of $f_1$ and $f_2$ are in $\Fsnip$, then $\Hsnip$ includes exactly one of
$e_1$ or $e_2$.
In turn, $\gammasnip$ goes through the two copies of $e$ an odd number of times total, meaning
$\gamma$ uses $e$ an odd number of times and $e \in H$.

In all cases, an edge $e$ is in $H$ if and only if exactly one incident face to $e$ is in $F$.
\end{proof}

Finally, we can present our algorithm for finding a minimum-weight separating subgraph that happens
to be contractible.
\begin{lemma}
\label{lem:contractible-alg}
There exists an $O(n \log \log n)$-time algorithm that computes a minimum-weight separating subgraph if any such subgraph is a simple contractible cycle. If not, the algorithm either returns some separating subgraph (that may not be minimum weight) or nothing.
\end{lemma}

\begin{proof}


The algorithm computes a system $A$ of tight arcs anchored on $s^*$ in $O(n)$ time as described in
Section~\ref{sec:characterizing_crossings}.
Let $\Gsnip$ denote the planar graph $G \snip A$; this graph has $O(n)$ vertices.

Pick an arbitrary edge $e$ on one of the tight arcs $\alpha$, and let $e_1$ and $e_2$ be distinct
copies of $e$
in~$\Gsnip$.  Let $\gamma_1$ and $\gamma_2$ be the shortest simple cycles in the  subgraphs $\Gsnip
\backslash e_1$ and $\Gsnip \backslash e_2$, respectively.  Our algorithm computes both~$\gamma_1$
and~$\gamma_2$ in $O(n \log\log n)$ time using the algorithm of \L\c{a}cki and
Sankowski~\cite{ls-mcsc-11}. Note that graphs~$\Gsnip \backslash e_1$ and~$\Gsnip \backslash e_2$
may not contain any cycles. In this case,~$\Gsnip$ contains no simple cycles.
Lemma~\ref{lem:disjoint-tight-arc} implies~$G$ does not contains any simple contractible cycles to begin with and our algorithm returns nothing. 
For the rest of this section, we assume~$\gamma_1$ and~$\gamma_2$ are well defined.

\begin{figure}[ht]
\centering
\includegraphics[height=1.2in]{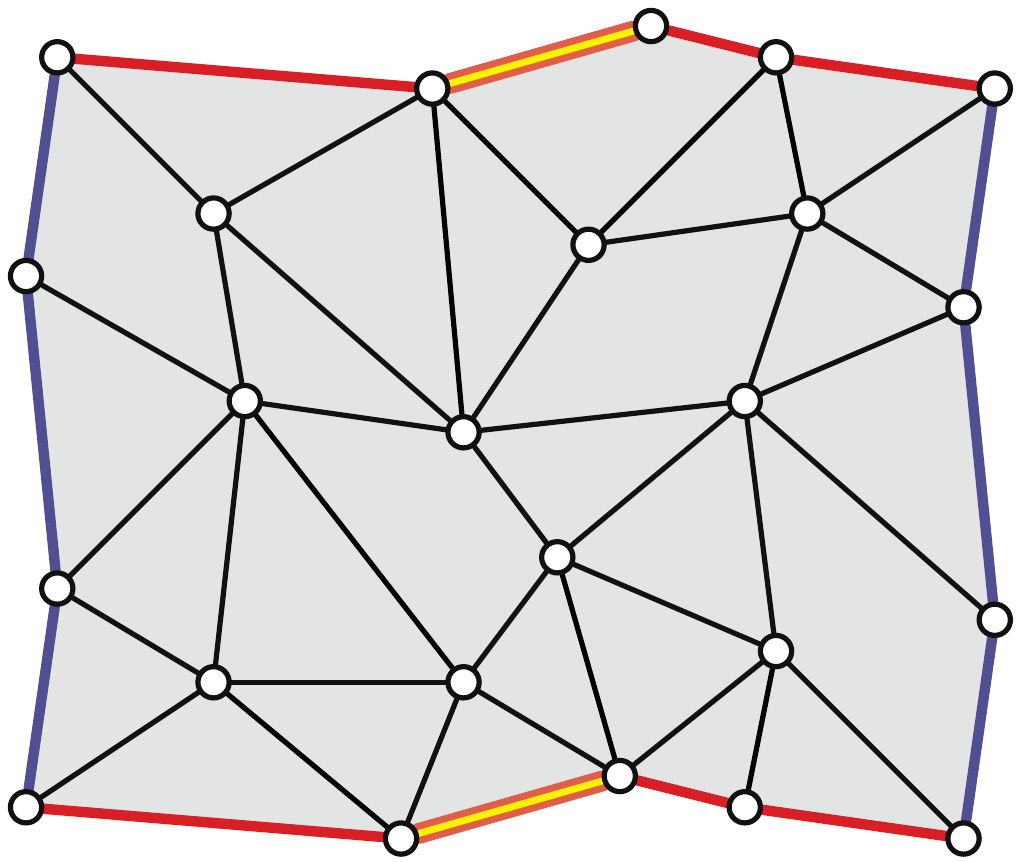}
\caption{At least one copy of $e$ is forbidden in the planarized graph.}
\label{fig:global_forbidden-pair}
\end{figure}

Let~$\gamma$ be the shorter of the cycles~$\gamma_1$ and $\gamma_2$.  By multiple instantiations of Lemma~\ref{lem:global_null-homologous-projections}, cycle~$\gamma$ projects to a null-homologous closed walk $\gamma'$ in the original graph $G$, which may or may not be simple.
Our algorithm returns the symmetric difference over all edges in~$\gamma'$.
The only cycle in $\Gsnip$ that is not a separating subgraph in~$G$ is the boundary of $\Gsnip$.
Because $\gamma$ avoids at least one edge of the boundary, the carrier of $\gamma'$ must be non-empty. If our algorithm returns anything, it must return a separating subgraph.

Now, suppose some minimum weight separating subgraph of~$G$ is a simple contractible cycle.
Lemma~\ref{lem:disjoint-tight-arc} implies that some shortest simple contractible cycle $\sigma$ in
$G$ does not cross $A$.  (We emphasize that our algorithm does not necessarily compute~$\sigma$.)  This cycle $\sigma$ appears as a simple cycle in $\Gsnip$ that avoids at least one of the edges $e_1$ or $e_2$.  Thus, $\sigma$ cannot be shorter than $\gamma$, and our algorithm returns a minimum-weight separating subgraph.
\end{proof}


\subsection{Non-contractible}
\label{sec:global_non-contractible}

Now suppose some minimum-weight separating subgraph $\minSS$ is non-contractible.  At a high level,
our algorithm for this case computes a set $F$ of faces, such that some minimum-weight separating
subgraph of $G$ separates $s^*$ from at least one face in $F$.  (Equivalently, $F^*$ is a set of
vertices of $G^*$, such that the global minimum cut in $G^*$ is an $(s,t)$-cut for some $t\in
F^*$.)  Then for each face in~$F$, we compute a minimum-weight subgraph separating $s^*$ from that
face using one of our earlier algorithms.

Throughout this section, we assume without loss of generality that every edge of $G$ lies on the boundary of two distinct faces of $G$.  We can enforce this assumption if necessary by adding $O(n)$ infinite-weight edges to $G$.

The following lemma can be seen as the main technical take-away from this section.
After its appearance in a preliminary version of our work~\cite{efn-gmcse-12}, it was
generalized by Borradaile \etal~\cite{benw-amcnt-16} for their construction of a minimum $(s,t)$-cut
oracle for surface embedded graphs.

\begin{lemma}
\label{lem:global_split-nocross}
Let~$\minSS$ be a minimum-weight separating subgraph.
Let~$\gamma$ be a closed walk in~$G$ that lies in the closure of the component of~$\Sigma \setminus
\minSS$ not incident to~$s^*$, and let~$H$
be a shortest even subgraph homologous to~$\gamma$.
There is a minimum weight separating subgraph~$\minSS'$ (possibly~$\minSS$) such
that~$H$ lies in the closure of the component of~$\Sigma \setminus \minSS'$ not incident to $s^*$. 
\end{lemma}

\begin{proof}
If $\gamma$ is null-homologous, then $H$ is empty and the lemma is trivial, so assume otherwise.
If~$H$ lies in the closure of the component of~$\Sigma \setminus \minSS$ not incident to $s^*$, then we are done, so assume otherwise.  See Figure \ref{fig:global_nonsep-vs-shortsep}. 

\begin{figure}[ht]
\centering
\includegraphics[height=1.25in]{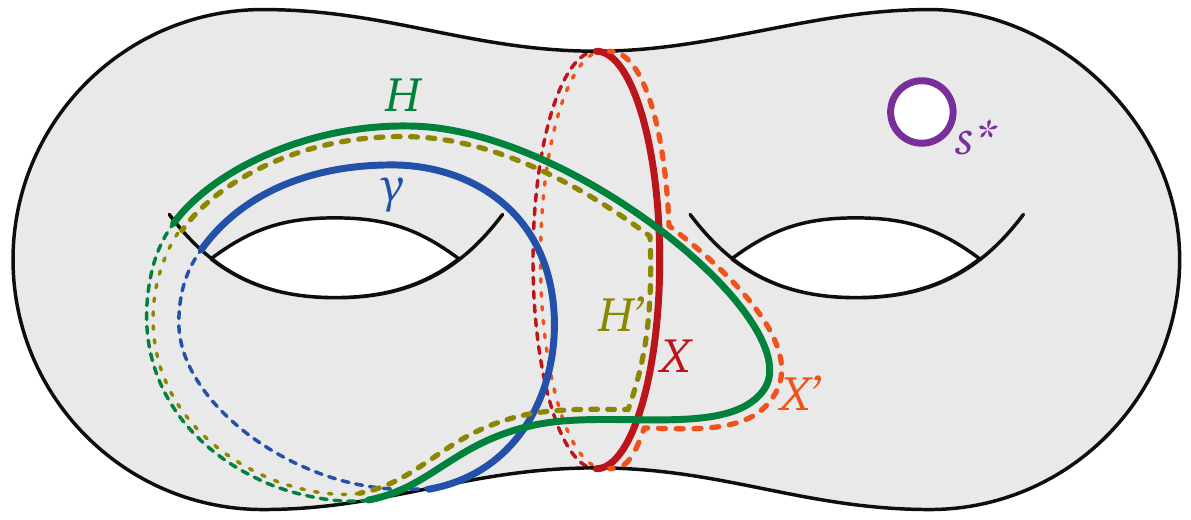}
\caption{The setting of Lemma~\ref{lem:global_split-nocross}. A~$\Z_2$-minimal even subgraph~$H$ is separated from face~$f$ by a minimum weight separating subgraph~$\minSS'$.}
\label{fig:global_nonsep-vs-shortsep}
\end{figure}

Recall, subgraph~$\minSS$ bounds the union of one non-empty component of faces not incident to
$s^*$.
Call the faces in this component \emph{far} and the rest \emph{near}.
Similarly, the even subgraph $H \oplus \gamma$ is null-homologous and therefore bounds a subset of
faces of $G$.
Call the faces in this subset \emph{white} and the rest \emph{black}.
(If $H = \gamma$, then every face of $G$ is black.)

Let~$\minSS'$ be the boundary of the union of the far faces and white faces in $G$.
There is at least one far face, so subgraph~$\minSS'$ is separating.
Every edge of~$H$ is incident to a white face, so~$H$ lies in the closure of the  component
of~$\Sigma \setminus \minSS'$ not incident to $s^*$.

It remains to argue that~$\minSS'$ is a minimum-weight separating subgraph of $G$.

For any subgraph~$A$ of $G$, let $w(A)$ denote the sum of the weights of the edges of $A$. Because both~$\minSS'$ and $\minSS$ are null-homologous, the even subgraph $H' = H \oplus \minSS' \oplus \minSS$ is homologous to $H$, and therefore to $\gamma$.  We immediately have $w(H') \ge w(H)$, because $H$ is $\Z_2$-minimal.

  We now prove that $w(\minSS') + w(H') \leq w(H) + w(\minSS)$ by bounding the contribution of each edge $e \in E(G)$ to both sides of the inequality.  Both $\minSS'$ and $H'$ are subgraphs of $\minSS\cup H$; moreover, $\minSS' \oplus H' = \minSS \oplus H$.  There are three cases to consider.
\begin{itemize}
\item
If $e \not\in \minSS\cup H$, then $e$ contributes $0$ to both sides of the inequality.
\item
If $e \in \minSS \oplus H$, then $e \in \minSS' \oplus H'$.  In this case, $e$ contributes $w(e)$ to both sides of the inequality.
\item
If $e \in \minSS \cap H$, then $e$ contributes exactly $2w(e)$ to the right side of the inequality.  Trivially, $e$ contributes at most $2w(e)$ to the left side.
\end{itemize}
We conclude that $\minSS'$ is also a minimum-weight separating subgraph.
\end{proof}

\begin{lemma}
\label{lem:global_split-alg}
There is a~$g^{O(g)} n \log \log n$-time algorithm that computes a minimum-weight separating subgraph of $G$ if any minimum-weight separating subgraph of $G$ is non-contractible.  If every minimum-weight separating subgraph of $G$ is contractible, the algorithm returns a separating subgraph that may not have minimum weight.
\end{lemma}

\begin{proof}
In a preprocessing phase, we  construct a homology basis from a tree-cotree decomposition in $O(gn)$ time~\cite{e-dgteg-03}.  Then we enumerate all~$2^{2g}-1$ non-trivial homology classes by considering subsets of cycles in this homology basis.  For each non-trivial homology class $h$, we perform the following steps:
\begin{itemize}
\item
Compute a minimum-weight subgraph $H_h$ in homology class $h$, in $g^{O(g)}n\log\log n$ time, as described by Theorem \ref{Th:Z2-minimal-crossing}.
\item
Fix an arbitrary edge $e$ of $H_h$.  By assumption, $e$ lies on the boundary of two distinct faces
$f_L$ and~$f_R$.  In particular, at least one of these faces is not $s^*$.
\item
If $f_L\ne s^*$, compute a minimum-weight subgraph $\minSS_h$ of $G$ that separates $s^*$ and $f_L$, in $g^{O(g)}n\log\log n$ time, using the minimum $(s,t)$-cut algorithm of Section~\ref{sec:crossing}.  Otherwise, $\minSS_h$ is undefined.
\item
If $f_R\ne s^*$, compute a minimum-weight subgraph $\minSS'_h$ of $G$ that separates $s^*$ and $f_R$, in $g^{O(g)}n\log\log n$ time, again using the minimum $(s,t)$-cut algorithm of Section~\ref{sec:crossing}.  Otherwise, $\minSS'_h$ is undefined.
\end{itemize}
Altogether we compute between $2^{2g}-1$ and $2^{2g+1}-2$ separating subgraphs of $G$ (some of which may coincide); the output of our algorithm is the smallest of these separating subgraphs.  The overall running time of our algorithm is $2^{O(g)} \cdot g^{O(g)}n\log\log n = g^{O(g)}n\log\log n$.

It remains to prove that our algorithm is correct.  Let~$\minSS$ be any minimum-weight subgraph
of~$G$ such that the component of $\Sigma \setminus \minSS$ not incident to $s^*$ is not a disk.
Let $\Sigma'$ be the closure of the component of $\Sigma \setminus \minSS$ that does not contain
$s^*$.  Because $\Sigma'$ is not a disk, it contains a cycle $\gamma$ that is not separating in
$\Sigma'$, and therefore not separating in $\Sigma$.  Let~$h$ be the homology class of $\gamma$ in
$\Sigma$, and let $H_h$ be any minimum-weight even subgraph of $G$ that is homologous with $\gamma$
in $\Sigma$.  Lemma~\ref{lem:global_split-nocross} implies, without loss of generality, that~$H_h$
lies entirely in~$\Sigma'$.  Thus, every edge of $H_h$ is on the boundary of at least one face $f'$
in~$\Sigma'$; it follow that $\minSS$ is a minimum-weight even subgraph separating $s^*$ and $f'$.  We conclude that when our algorithm considers homology class $h$, either $\minSS_h$ or $\minSS'_h$ is a minimum-weight separating subgraph of $G$.
\end{proof}

Modifying the previous algorithm to use results of Section \ref{sec:homcover}, instead of the corresponding results in Section \ref{sec:crossing}, immediately gives us the following:

\begin{lemma}
\label{lem:global_split-alg2}
There is a~$2^{O(g)} n \log n$-time algorithm that computes a minimum-weight separating subgraph of $G$ if any minimum-weight separating subgraph of $G$ is non-contractible.  If every minimum-weight separating subgraph of $G$ is contractible, the algorithm returns a separating subgraph that may not have minimum weight.
\end{lemma}

\subsection{Summing up}

Finally, to compute the minimum-weight separating subgraph in $G$, we run both  algorithms described in Lemmas~\ref{lem:contractible-alg} and~\ref{lem:global_split-alg}.  If either algorithm returns nothing, the other algorithm returns a minimum-weight separating subgraph of $G$; otherwise, both algorithms return non-empty separating subgraphs of $G$, and the smaller of those two subgraphs is a minimum-weight separating subgraph of $G$.  We conclude:

\begin{theorem}
Let $G$ be an edge-weighted undirected graph embedded on a surface with genus $g$ with one boundary
component.
We can compute a minimum-weight separating subgraph in $G$ in either $g^{O(g)} n \log \log n$ time or $2^{O(g)} n \log n$ time.
\end{theorem}

\begin{corollary}
Let $G$ be an edge-weighted undirected graph embedded on a surface with genus $g$.
We can compute a global minimum cut in $G$ in either $g^{O(g)} n \log \log n$ time or $2^{O(g)} n \log n$ time.
\end{corollary}

%

\paragraph{Acknowledgments}
The authors would like to thank Chandra Chekuri and Aparna Sundar for helpful discussions on some of the preliminary work included here. 
We would also like to thank the anonymous reviewers of our earlier extended abstracts \cite{cen-mcshc-09, en-mcsnc-11, efn-gmcse-12} for many helpful comments and suggestions.

\bibliographystyle{abuser}
\bibliography{bib/topology,bib/data-structures,bib/optimization,bib/other}

\end{document}